\documentclass[12pt,a4paper,headings=small,oneside]{scrartcl}

\usepackage{graphicx}
\usepackage{amsmath,amssymb,amsthm}
\usepackage{url}
\usepackage{ifthen}
\usepackage{natbib}
\usepackage{color}
\usepackage{subfigure}
\usepackage{epstopdf}
\usepackage{arydshln}
\usepackage[colorlinks, linkcolor = black, citecolor = black, filecolor = black, urlcolor = blue]{hyperref}
\usepackage{soul}
\theoremstyle{plain}
\newtheorem{theorem}{Theorem}%
\newtheorem{proposition}[theorem]{Proposition}%
\newtheorem{lemma}[theorem]{Lemma}%
\theoremstyle{definition}
\newtheorem{remark}{Remark}%

\usepackage{marginnote}

\usepackage[colorinlistoftodos,textsize=tiny]{todonotes}
\newcommand{\Comments}{1}
\newcommand{\mynote}[2]{\ifnum\Comments=1\textcolor{#1}{#2}\fi}
\newcommand{\mytodo}[2]{\ifnum\Comments=1%
	\todo[linecolor=#1!80!black,backgroundcolor=#1,bordercolor=#1!80!black]{#2}\fi}

\ifnum\Comments=1               
\fi

\renewcommand\theta{\vartheta}

\newcommand{\R}{\ensuremath{\mathbb{R}}}%

\newcommand{\rhom}{\ensuremath{\rho_{\text{max}}}}%
\newcommand{\hrace}{\ensuremath{\widehat{\rho_{\text{ace}}}}}%
\newcommand{\rpear}{\ensuremath{\rho_{\text{Pear}}}}%
\newcommand{\hrpear}{\ensuremath{\widehat{\rho_{\text{Pear}}}}}%
\newcommand{\rspear}{\ensuremath{\rho_{\text{S}}}}%
\newcommand{\hrspear}{\ensuremath{\widehat{\rho_{\text{S}}}}}%

\newcommand{\hdcor}{\ensuremath{\widehat{\text{dCor}}}}%
\newcommand{\hxi}{\ensuremath{\widehat{\xi}}}%
\newcommand{\tstar}{\ensuremath{t^*}}%
\newcommand{\hhsic}{\ensuremath{\widehat{\text{hsic}}}}%
\newcommand{\hbcor}{\ensuremath{\widehat{\text{bCor}}}}%

\newcommand{\rhoamn}{\ensuremath{\rho_{\text{L,l}}}}%
\newcommand{\rhono}{\ensuremath{\rho_{\text{l1}}}}%
\newcommand{\rhont}{\ensuremath{\rho_{\text{l2}}}}%

\newcommand{\hrhoamn}{\ensuremath{\widehat{\rho_{\text{L,l}}}}}%
\newcommand{\hrhono}{\ensuremath{\widehat{\rho_{\text{l1}}}}}%
\newcommand{\hrhont}{\ensuremath{\widehat{\rho_{\text{l2}}}}}%

\newcommand{\rhoam}{\ensuremath{\rho_{\text{L}}}}%
\newcommand{\rhoro}{\ensuremath{\rho_{\text{R1}}}}%
\newcommand{\rhort}{\ensuremath{\rho_{\text{R2}}}}%

\newcommand{\hrhoam}{\ensuremath{\widehat{\rho_{\text{L}}}}}%
\newcommand{\hrhoro}{\ensuremath{\widehat{\rho_{\text{R1}}}}}%
\newcommand{\hrhort}{\ensuremath{\widehat{\rho_{\text{R2}}}}}%

\newcommand{\hrho}{\ensuremath{\hat{\rho}}}

\newcommand{\rranksn}{\ensuremath{\rho_{\text{rk,2,n}}}}%

\newcommand{\cL}{\ensuremath{\mathcal{L}}}%
\newcommand{\Ex}{\ensuremath{\mathbb{E}}}%
\newcommand{\PP}{\ensuremath{\mathbb{P}}}%

\newcommand{\cX}{\ensuremath{\mathcal{X}}}%
\newcommand{\cY}{\ensuremath{\mathcal{Y}}}%

\newcommand{\dd}{\ensuremath{{\mathrm d}}}%

\newcommand{\D}{\,\mbox{\textrm D}}
\newcommand{\Cov}{\ensuremath{\text{Cov}}}%
\newcommand{\ind}{{\boldsymbol 1}}

\DeclareMathOperator{\Var}{Var}

\newboolean{comment}
\newcommand{\Kommentar}[1]{%
	\ifthenelse{\boolean{comment}}{#1}{}
}
\setboolean{comment}{true}

\begin{document}
\title{Lancaster correlation - a new dependence measure linked to maximum correlation}


\author{Hajo Holzmann$^{1}$ and Bernhard Klar\,$^{2}$   \\
	        \small{$^{1}$ Fachbereich Mathematik und Informatik, Philipps-Universit\"at Marburg, Germany}  \\
	        \small{$^{2}$ Institut f\"ur Stochastik, Karlsruher Institut f\"ur Technologie (KIT), Germany} \\
	        \small{holzmann@mathematik.uni-marburg.de, bernhard.klar@kit.edu}}

\date{\today }
\maketitle

\begin{abstract}
We suggest novel correlation coefficients which equal the maximum correlation for a class of bivariate Lancaster distributions while being only slightly smaller than maximum correlation for a variety of further bivariate distributions. In contrast to maximum correlation, however, our correlation coefficients allow for rank and moment-based estimators which are simple to compute and have tractable asymptotic distributions. Confidence intervals resulting from these asymptotic approximations and the covariance bootstrap show good finite-sample coverage. In a simulation, the power of asymptotic as well as permutation tests for independence based on our correlation measures compares favorably with competing methods based on distance correlation or rank coefficients for functional dependence, among others. Moreover, for the bivariate normal distribution, our correlation coefficients equal the absolute value of the Pearson correlation, an attractive feature for practitioners which is not shared by various competitors.  We illustrate the practical usefulness of our methods in applications to two real data sets.
\end{abstract}


\section{Introduction}

Measuring the correlation between random quantities is a basic task for statisticians. The most commonly used correlation coefficients for real-valued observations are Pearson's correlation coefficient as well as Spearman's $\rho$ and Kendell's $\tau$, designed for detecting linear and monotone relationships, respectively. Still, additional requirements such as characterizing independence also in multiple dimensions \citep{szekely2007measuring, zbMATH07072331} or even in abstract Hilbert spaces \citep{pan:2020}, regression dependence \citep{dette2013copula, wang2017generalized, chatterjee2021new}, or consistency axioms \citep{bergsma2014consistent, weihs2018symmetric}, among others, led to a variety of novel correlation coefficients.

In terms of axiomatic properties, the benchmark may be considered to be the maximum correlation introduced by \citet{hirschfeld1935connection} and \citet{gebelein1941statistische}. It is defined for non-constant random variables $X, Y$ on the same probability space
 with values in potentially abstract measurable spaces $\cX$, $\cY$ by
\begin{equation}\label{eq:maxcor}
 \rhom(X,Y) = \sup_{\phi,\psi} \, \Ex\big[\phi(X)\, \psi(Y)\big],
 \end{equation}
where the supremum is taken over
\begin{equation}\label{eq:normalized}
	\phi:\cX \to \R, \ \Ex[\phi(X)] = 0, \  \Ex[ \phi^2(X)]=1, \quad
\end{equation}
%
and similarly for $\psi$ and $Y$. \citet{renyi1959measures} sets up a list of axiomatic properties for correlation coefficients for real-valued quantities, including symmetry, characterization of independence and functional dependence, and invariance under monotone transformations, and shows that these are satisfied by maximum correlation. 
Further, for the normal distribution it equals the absolute value of Pearson's correlation coefficient \citep{gebelein1941statistische}, see also \citet{lancaster1957some}.

While maximum correlation is still of broad theoretical and applied interest \citep{lopez2013randomized, anantharam2013maximal, huang2016model, baharlouei2019r}, there are some downsides with its involved definition as a supremum in \eqref{eq:maxcor}.
First, it is hard to estimate, and while the alternating conditional expectations (ACE) - algorithm of \citet{breiman1985estimating} is a well-established tool to do so, it gives only raw estimates so that tests for independence or confidence intervals are not readily available.
Furthermore, maximum correlation is sometimes considered to be close to or even equal to $1$ for too many joint distributions \citep{pregibon1985estimating, chatterjee2021new}.

We propose novel correlation coefficients which are motivated by a simple expression of the maximum correlation for the class of bivariate Lancaster distributions. In contrast to various other recent proposals, for the bivariate normal distribution, our correlation coefficients thus equal the absolute value of the Pearson correlation. 
They rely on the maximum of the linear correlation between the original quantities on the one hand, and on that of their squares on the other hand, and thus are easily visualized and interpretable. As we shall illustrate, our correlation coefficients equal or are only slightly less than maximum correlation for a variety of further bivariate distributions. In contrast to maximum correlation, however, these novel coefficients allow for rank - and moment-based estimators which are simple to compute and have tractable asymptotic distributions. These result in asymptotic confidence intervals with good coverage probabilities which are not readily available for maximum correlation itself nor for various other competitors. While other correlation measures such as distance correlation or the coefficient by \citet{chatterjee2021new} are better suited for discriminating regression-type dependence from independence, our simulations show that independence tests based on our coefficients have higher power against various other forms of dependence. Thus, our coefficients can be used complementary to these correlation measures.


More formally suppose that $X$ and $Y$ are marginally standard normally distributed with joint density $f_{X,Y}$. \citet{sarmanov1967probabilistic} characterize the densities $f_{X,Y}$ which are of finite squared contingency in the sense of \citet{pearson1904theory} and admit diagonal expansions in terms of Hermite polynomials, see also \citet{lancaster1958structure, lancaster1963correlations}. Their results imply that the maximum correlation is then given as
\begin{equation}\label{eq:expressmaxcor0}
	\rhom(X,Y) = \max\big(|\rpear(X,Y)|, \rpear(X^2,Y^2) \big),
\end{equation}
where $\rpear$ is the Pearson correlation coefficient. For the bivariate normal, this reduces to $|\rpear(X,Y)|$.
The expansions in terms of canonical forms on which \eqref{eq:expressmaxcor0} is based give rise to the class of so-called Lancaster distributions, for recent results beyond normal margins see \citet{buja1990remarks, koudou1998lancaster, papadatos2013simple}.

The representation \eqref{eq:expressmaxcor0} is the formal motivation for introducing our correlation coefficient, which we shall call the \textit{Lancaster correlation coefficient}. We set
\begin{equation}\label{eq:expressmaxcor}
    \rhoam(X,Y) = \max\big(|\rpear(\tilde X, \tilde Y)|, \big|\rpear(\tilde X^2,\tilde Y^2) \big|\big),\quad
    \tilde X = \Phi^{-1}(F_X(X)), 
\end{equation}
$\tilde Y$ is defined analogously to $\tilde X$, and $\Phi$ denotes the distribution function of the standard normal and $F_X$ and $F_Y$ are the distribution functions of $X$ and $Y$, respectively.
If $X$ and $Y$ have continuous marginal distributions such that the distribution of $(\tilde X, \tilde Y)$ falls into the class discussed in \citet{sarmanov1967probabilistic}, then actually $\rhoam(X,Y) = \rhom(X,Y)$, otherwise, $\rhoam(X,Y)$ is a lower bound for $\rhom(X,Y)$.
For general distributions, $\rhoam(X,Y)$ measures the linear correlation between $\tilde X$ and $\tilde Y$ as well as between their squares. As we shall illustrate numerically this suffices to capture most of $\rhom(X,Y)$ for a variety of distributions. Also note that $\rpear(\tilde X, \tilde Y)$ is closely related to the Spearman correlation coefficient, with scores from the normal quantile function in the spirit of the van der Waerden test \citep{van1952order}. Such transformations are known under various names, e.g. rank-based inverse normal in psychology \citep{beasley2009,bishara2017}.

The structure of the paper is as follows. In Section \ref{sec2}, after reviewing some further details leading to the expression \eqref{eq:expressmaxcor0}, we propose a second correlation coefficient with a simple moment standardization. Next in a motivating simulation, we consider various bivariate distributions which capture distinct forms of dependence and compute estimates of several correlation coefficients including our novel proposals. Thereby we illustrate the potential of these correlation coefficients to detect dependence as well as the fact that various competing methods have drastically smaller values than Pearson correlation for the bivariate normal distribution.
Furthermore, we show that our proposals give estimates close to those of maximal correlation. Here we use rank - and moments-based estimators, which are then formally introduced in Section \ref{sec3}. We derive the asymptotic distribution of our estimators and discuss how to construct confidence intervals and tests for independence. Let us stress that asymptotic confidence intervals are not readily available for maximum correlation nor for several further competitors. In Section \ref{sec4} we conduct a simulation study. First in Section \ref{sec-asym-ci} numerical experiments show good coverage properties of confidence intervals based on asymptotics and the covariance bootstrap. Next in Section \ref{sec-ind-test} we show that the power of asymptotic as well as permutation tests for independence based on our correlation measures compares favorably in simulations to various competitors except for regression-type dependence. These competitors include maximum correlation itself estimated with the ACE algorithm of \citet{breiman1985estimating}, distance correlation \citep{szekely2007measuring},  the rank coefficient for functional dependence by \citet{chatterjee2021new}, the $\tau^*$ - coefficient of \citet{bergsma2014consistent}, ball covariance \citep{pan:2020} or reproducing kernel Hilbert space methods \citep{zbMATH05070116}. In Section \ref{sec:dataex}, the practical usefulness of our correlation coefficients to detect dependence is illustrated on two data sets, one from \cite{fox2019}, Sec. 10.7.1, containing the 2008-09 nine-month academic salary for Professors in a U.S.~college, the other consisting of pairs $(r_t,r_{t-1})$ of hourly observations of log-returns of the three cryptocurrencies Bitcoin (BTC), Ethereum (ETH) and Ripple (XRP) from May 16, 2018 to October 27, 2021, as considered in \citet{holzmann2023} in a forecasting context. Section \ref{sec:conclude} concludes, while the appendix contains technical proofs of the main results. Some further technical details and simulation results are provided in the supplementary material.
An R package \verb+lancor+\footnote{\href{https://github.com/BernhardKlar/lancor}{https://github.com/BernhardKlar/lancor}} containing functions to compute the coefficients together with confidence intervals, and to perform permutation tests, is available on \verb+github+.




\section{Maximum correlation and the Lancaster correlation coefficient} \label{sec2}

%

In this section, we further motivate and illustrate our novel correlation coefficient $\rhoam(X,Y)$ in \eqref{eq:expressmaxcor}. We start by recalling the nonparametric class of densities for which \eqref{eq:expressmaxcor0}, our point of departure, holds true. Let $\varphi(x)$ denote the density of the standard normal distribution, and suppose that $X$ and $Y$ have standard normal margins. The density $f_{X,Y}$ of $(X,Y)$ is by definition of finite squared contingency if and only if $f_{X,Y}(x,y)/(\varphi(x)\, \varphi(y))$ is square - integrable in the $L_2$-space with the bivariate standard Gaussian measure. Then if the canonical variables, that is the singular functions of the conditional expectation operators, are the Hermite polynomials $H_k$, $f_{X,Y}$ has an expansion
\begin{equation}\label{eq:Lancasterexpansion}
f_{X,Y}(x,y) = \varphi(x)\, \varphi(y)\, \Big(1 + \sum_{k=1} c_k\, H_k(x)\, H_k(y)\Big),
\end{equation}
see \citet[Theorem 3]{lancaster1958structure}. 
Here, $c_k$ is the Pearson correlation coefficient of $H_k(X)$ and $H_k(Y)$. Conversely, \citet{sarmanov1967probabilistic} show that \eqref{eq:Lancasterexpansion} defines a density if and only if $c_k = \Ex[V^k]$ is the moment sequence of some random variable $V$ with support in $(-1,1)$. By expanding the functions $\phi$ and $\psi$ in \eqref{eq:normalized} in the definition of the maximum correlation into Hermite series, one finds that $\rhom(X,Y)$ equals the supremum over $|c_k|$, and from the representation of $c_k$ as moments of $V$ this results in $\max(|c_1|,c_2)$, that is, \eqref{eq:expressmaxcor0}.

Before we proceed to numerical illustrations on a variety of distributions, as an alternative to \eqref{eq:expressmaxcor} let us also introduce the \textit{linear Lancaster correlation coefficient} as
\begin{equation}\label{eq:expressmaxcor1}
	\rhoamn(X,Y) = \max\big(|\rpear(X, Y)|, |\rpear(\breve X^2,\breve Y^2)| \big),
\end{equation}
where $\breve X = (X - \Ex[X])/\text{sd}(X)$ and $\breve Y$ are the standardized versions of $X$ and $Y$ and we assume $\Ex[X^4]<\infty$, $\Ex[Y^4]<\infty$.

\medskip
To illustrate, we first consider bivariate normal distributions with correlation coefficient $\rho=0.1$ and 0.3, called BVN(0.1) and BVN(0.5).
Second, we consider bivariate normal mixtures $p f_1(x,y,-1/2) + (1-p)$ $f_1(x,y,1/2)$, where $f_1(x,y,\rho)$ denotes the density of a bivariate normal distribution with standard normal marginals and correlation coefficient $\rho$.
For $p=1/2$, one has $\rpear=0,\, \rhom=1/4$, called NM1. Distribution NM2 uses $p=1/3$, resulting in $\rpear=1/6,\, \rhom=1/4$
\citep{sarmanov1967probabilistic}.
Scatterplots of the two distributions based on a sample of size 10000 are shown in the left and middle panel of Fig. \ref{fig1:example}.

\begin{figure}
    \centering
    \includegraphics[width=0.95\textwidth]{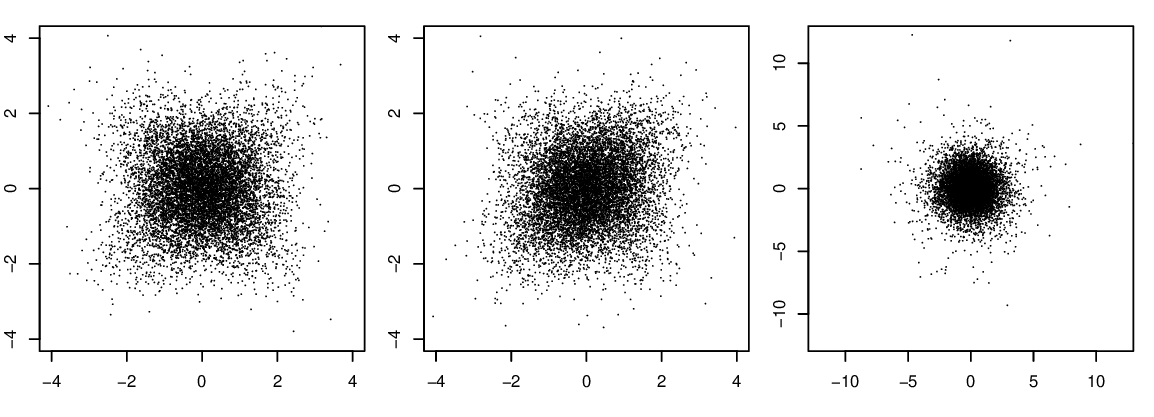}
    \caption{Scatterplot with sample size $n=10000$ of two bivariate normal mixtures (NM1, NM2) and bivariate $t$-distribution with 5 degrees of freedom (BVT5)}
    \label{fig1:example}
\end{figure}

Estimates of various dependence measures for these datasets can be found in Table \ref{dist-tab1}. Besides $\hrpear$ and $\hrspear$, the empirical versions of Pearson's and Spearman's correlation coefficients, we use $\hrhoamn$ and $\hrhoam$, defined in Section \ref{sec3}, as estimators of $\rhoamn$ and $\rhoam$. Further, $\hrace$, the estimate of $\rhom$ computed by the ACE algorithm in the R package \verb+acepack+ 
\citep{acepack2016}; the empirical distance correlation $\hdcor$ \citep{energy2022}; the empirical value $t^*$ of the $\tau^*$ coefficient of \cite{bergsma2014consistent}, computed with the R package \verb+TauStar+ \citep{taustar2019}; $\hxi$, the empirical version of the $\xi$-coefficient of \cite{chatterjee2021new}; finally, the empirical value $\hbcor$ of the ball correlation coefficient of \cite{pan:2020}, computed with the R package \verb+Ball+ \citep{Zhu:2021}.

{\bfseries Findings:} 
As expected, the values of $\hrpear, \hrhoamn$ and $\hrhoam$ coincide for BVN(0.1) and BVN(0.3), and $\hrace$ is nearly identical. Spearman's $\hrspear$ is a bit smaller under normality, which also holds for $\hdcor$. The values of $t^*,\hxi$ and $\hbcor$ are close to 0 for $\rho=0.1$, and still quite small for $\rho=0.3$. In particular, the ratio $\rho/\hbcor$ is around 100.]

For NM1, the values of $\hrhoamn, \hrhoam$ and $\hrace$ are close to $\rhom=1/4$, as expected. Whereas $\hdcor$ is quite small, $t^*$, $\hxi$ and $\hbcor$ do not detect the dependence at all.
For NM2, $\hrpear,\hrspear,\hrace$ and $\hdcor$ estimate the linear dependence 1/6; only $\hrhoamn$ and $\hrhoam$ have estimates close to $\rhom=1/4$. Again, $t^*$ and $\hxi$ do not detect the dependence.

\begin{table}
\centering
\begin{tabular}{lccccccccc}
  \hline
distribution & \hrpear & \hrspear & \hrhoamn & \hrhoam & \hrace & \hdcor & \tstar & \hxi & \hbcor \\ 
  \hline
BVN(0.1)   &  0.092 & 0.089 & 0.092 & 0.092 & 0.097 & 0.085 & 0.003 & 0.011 & 0.001 \\ 
BVN(0.3)   &  0.290 & 0.274 & 0.290 & 0.290 & 0.291 & 0.256 & 0.025 & 0.055 & 0.005 \\ 
NM1        &  0.002 &-0.006 & 0.279 & 0.279 & 0.256 & 0.062 & 0.001 & 0.005 & 0.002 \\ 
NM2        &  0.154 & 0.145 & 0.237 & 0.238 & 0.155 & 0.147 & 0.008 & 0.014 & 0.004 \\ 
BVT5       & -0.010 &-0.000 & 0.190 & 0.210 & 0.210 & 0.060 & 0.000 & 0.000 & 0.000 \\ 
BVC        &  0.193 & 0.004 & 0.968 & 0.718 & 0.992 & 0.780 & 0.006 & 0.046 & 0.040 \\ 
Unif-disc  &  0.002 & 0.002 & 0.335 & 0.267 & 0.334 & 0.076 & 0.002 & 0.022 & 0.004 \\ 
GARCH(2,1) &  0.021 & 0.020 & 0.563 & 0.516 & 0.559 & 0.139 & 0.002 & 0.032 & 0.010 \\    \hline
\end{tabular}
\caption{Estimates of correlation measures for several distributions, each based on one sample with size $n=10000$ \label{dist-tab1}}
\end{table}

The right panel of Fig. \ref{fig1:example} shows a scatterplot of the bivariate (standard) $t$-distributions with 5 degrees of freedom and linear correlation $\rpear=0$, called BVT5. A similar plot for the bivariate standard Cauchy distribution (BVC) and a uniform distribution on the unit disc (Unif-disc) is shown in the left and middle panel of Fig. \ref{fig2:example}. Pearson correlation exists and equals zero for the first and the third distribution; $\rspear=0$ in all three cases.
As is well known, the only independent spherically symmetric distributions are normal distributions. For Unif-disc, it is known that $\rhom=0.3$ \citep{buja1990remarks}.

\begin{figure}
    \centering
    \includegraphics[width=0.95\textwidth]{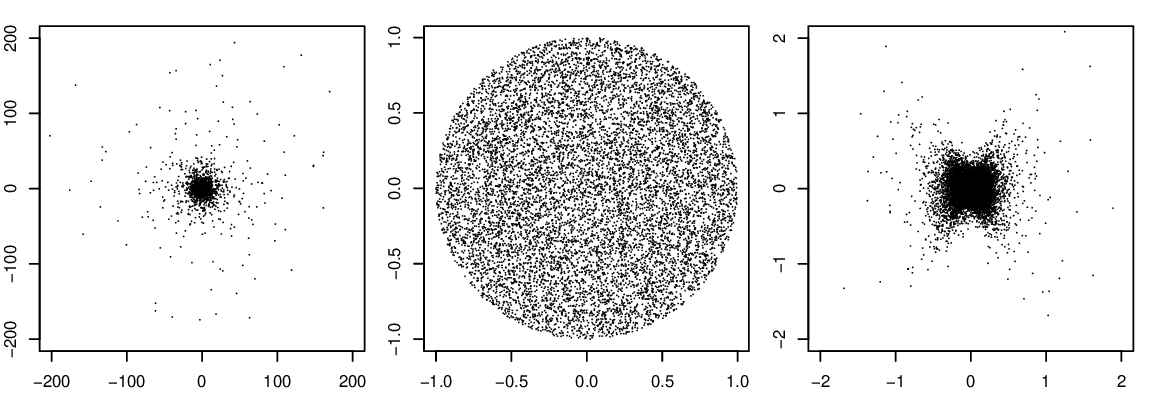}
    \caption{Scatterplot with sample size $n=10000$ of bivariate Cauchy (BVC), uniform distribution on the disc (Unif-disc) and GARCH(2,1)}
    \label{fig2:example}
\end{figure}

{\bfseries Findings:} Clearly, $\hrpear$ and $\hrspear$ can not capture the dependence for either of the three distributions; the same holds for $t^*$, $\hxi$ and $\hbcor$. For BVT5, $\hrhoamn,\hrhoam$ and $\hrace$ give similar values around 0.2; $\hdcor$ is rather small. For BVC, $\hrhoam$ and $\hrace$ take quite large values; $\hrpear, \hrhoamn$ and $\hdcor$ are not defined. For Unif-disc, $\hrhoamn, \hrhoam$ and $\hrace$ take values around 0.3, $\hdcor$ is again rather small.

\smallskip
Stock returns $r_t$ 
are typically uncorrelated with their own past $r_{t-1}$ (if stock markets are efficient), but their squares $r_t^2$ and $r_{t-1}^2$ are not. To illustrate the behaviour of our new correlation measures for a typical statistical model for returns, we generate $r_t$ using a GARCH(2,1) process with parameters $\alpha=(0.01,0.6)$ and $\beta=0.2$. It is well known that $\rpear=0$ holds for the distribution of the pairs $(r_t, r_{t-1})$.
A scatterplot of $(r_t, r_{t-1})$ is shown in the right panel of Fig. \ref{fig2:example}, estimates based on a sample of size 10000 are given in Table \ref{dist-tab1}.

{\bfseries Findings:} $\hrpear$ and $\hrspear$ are close to zero as expected, but the same holds for  $t^*$, $\hxi$ and $\hbcor$. While $\hdcor$ is around 0.15, $\hrhoamn, \hrhoam$ and $\hrace$ take values larger than 0.5.


\section{Estimation and asymptotic inference} \label{sec3}
Let us discuss estimation of $\rhoamn(X,Y)$ and of $\rhoam(X,Y)$ based on an i.i.d.~sample $(X_1, Y_1),$ $\ldots,$ $(X_n,Y_n)$ from the distribution of $(X,Y)$.

First consider $\rhoamn(X,Y)$, set $\rhono = \rpear(X, Y)$ and $\rhont = \rpear(\breve X^2, \breve Y^2)$, see \eqref{eq:expressmaxcor1}. Let $\hrhono$ be the sample correlation, and let $\hrhont$ be the empirical correlation of the squares of the empirically standardized observations. We set
\begin{equation}
	\hrhoamn = \max\big(|\hrhono|, |\hrhont| \big).
\end{equation}
\begin{proposition}\label{lem:asympnormest}
If $\Ex[X^8]<\infty$ and $	\Ex[Y^8]<\infty$ and $X$ and $Y$ are not almost surely constant, then
\begin{equation*} 
	\sqrt n \, \big((\hrhono, \hrhont) - ( \rhono,  \rhont) \big) 	\stackrel{\cL}{\to} \mathcal N_2\big(0, \Sigma^* \big).
\end{equation*}
For the asymptotic covariance matrix we have that $\Sigma^* = M \, \Sigma_m \, M^\top $ for a $12 \times 12$ covariance matrix $\Sigma_m$ of moments, and a $2 \times 12$ - matrix $M$ arising from the $\Delta$ - method, which are given in the appendix.
\end{proposition}
The proof is a routine application of the delta method and is provided in the supplementary material. 	
%
If $X$ and $Y$ are independent, the asymptotic covariance matrix reduces to
	\begin{align} \label{eq:asymcovmat-ind}
		\Sigma^* &= \begin{pmatrix}
			1 & \tau \\  \tau & 1
		\end{pmatrix},
		\qquad  \tau = \frac{ \Ex[\breve X^3]\,\Ex[\breve Y^3]  }{ \sqrt{(\Ex[\breve X^4]-1)(\Ex[\breve X^4]-1)} },
	\end{align}
	and hence to the unit matrix if one of the third moments vanishes, particularly for symmetric distributions.
 Under bivariate normality but without independence the asymptotic covariance matrix $\Sigma^*$ reduces to
\begin{align} \label{eq:asymcovmat-bvn}
	\Sigma^* &= \begin{pmatrix}
		(1-\rho^2)^2  & 2\rho(1-\rho^2)^2,  \\
		2\rho(1-\rho^2)^2 & (1-\rho^2)^2 (3\rho^4+10\rho^2+1)
	\end{pmatrix},
\end{align}
where $\rho$ is the Pearson correlation. The first diagonal entry is the well-known variance in the limiting normal distribution of  $\hrpear$ (see, e.g., \cite{lehmann1998}, Example 6.5). See the supplement for the computation leading to \eqref{eq:asymcovmat-bvn}.
%

Next consider rank-based estimation of $\rhoam(X,Y)$. Set $\rhoro = \rpear(\tilde X, \tilde Y)$ and $\rhort = \rpear(\tilde X^2, \tilde Y^2)$, see \eqref{eq:expressmaxcor}.
Let $Q_i=\sum_{j=1}^n \ind\{X_j\leq X_i \}, i=1,\ldots,n,$ denote the rank of $X_i$ within $X_1,\ldots,X_n$, and, likewise, let $R_i=\sum_{j=1}^n \ind\{Y_j\leq Y_i\}, i=1,\ldots,n,$ be the rank of $Y_i$. As estimators for $\rhoro$ and $\rhort$ we propose
\begin{align*}
	\hrhoro &= \frac{\sum_{j=1}^n a(Q(j))\, a(R(j))}{ns_a^2}, \qquad 	\hrhort = \frac{\sum_{j=1}^n \big(b(Q(j)) - \bar b\big)\, \big(b(R(j))- \bar b\big)}{ns_b^2},	
\end{align*}
where we take the scores
\begin{align*}
	a(j) & =\Phi^{-1}\left(\frac{j}{n+1}\right),\quad  j=1,\ldots,n,
	\qquad
	s_a^2 = \frac{1}{n} \sum_{j=1}^n \left(a(j)-\bar{a}\right)^2,\\
	b(j) & = a^2(j), \quad
	\bar{b} = \frac{1}{n} \sum_{j=1}^n b(j), \quad
	s_b^2 = \frac{1}{n} \sum_{j=1}^n \left(b(j)-\bar{b}\right)^2.
\end{align*}
Further, we set
$$\hrhoam = \max\big(|\hrhoro|, |\hrhort| \big).$$
%
\begin{proposition}\label{lem:jointasympnormalrankgeneral}
Assume that $X$ and $Y$ are continuously distributed. Then $\sqrt n \big((\hrhoro, \hrhort)$ $- (\rhoro, \rhort) \big)$ is asymptotically normally distributed. If furthermore, $X$ and $Y$ are independent, then $\sqrt n (\hrhoro, \hrhort)$ is asymptotically bivariate standard normally distributed.
%
%
\end{proposition}
The first statement follows from \citet[proof of Theorem 2.1]{ruymgaart1972asymptotic}, while the second part is more easily derived using antiranks and the rank central limit theorem under the Noether condition, see \citet[Section 13.3]{van2000asymptotic}. We provide the details in the appendix.







From Propositions \ref{lem:asympnormest} and \ref{lem:jointasympnormalrankgeneral} we may now deduce the asymptotic distributions of
$\hrhoam$ and $\hrhoamn$. To this end let $\rho_j$ and $\hrho_j$, $j=1,2,$ denote either $\rho_{\text{Rj}}$ and $\widehat{\rho_{\text{Rj}}}$ or $\rho_{\text{lj}}$ and $\widehat{\rho_{\text{lj}}}$, so that we have established that
\begin{equation}\label{eq:asympcoeff}
	\sqrt n \, \big((\hrho_1, \hrho_2) - ( \rho_1,  \rho_2) \big) 	\stackrel{\cL}{\to} (U,V) \sim \mathcal N_2\big(0, \Sigma \big),
\end{equation}
and need to deduce the asymptotic distribution of
\begin{equation}\label{eq:finalasymptotic}
	\sqrt n \, \big(\max(|\hat \rho_1|, |\hat \rho_2|) - \max(|\rho_1|, |\rho_2|) \big)
\end{equation}
\begin{theorem}\label{lem:asympdistrmax}
Assume \eqref{eq:asympcoeff}.
\begin{enumerate}
	\item If $|\rho_1|> |\rho_2|$, the asymptotic distribution in \eqref{eq:finalasymptotic} is $U$, while if $|\rho_1|< |\rho_2|$ it is $V$.

	%
	\item If $\rho_1 = \rho_2 \not = 0$, the asymptotic distribution in \eqref{eq:finalasymptotic} is $\max(U,V)$, while if $\rho_1 = - \rho_2 \not = 0$ it is $\max (-U,V)$. See Lemma \ref{lem:maxnormal} below for the distribution function and density of $\max(U,V)$ and $\max (-U,V)$.
		%
	\item If $\rho_1 = \rho_2 = 0$, the asymptotic distribution in \eqref{eq:finalasymptotic} is $\max (|U|,|V|)$. See Lemma \ref{lem:maxnormal1} in the appendix for the distribution and density functions. If $(U,V)$ in \eqref{eq:asympcoeff} are bivariate standard normally distributed, which is the case for $\hrhoam$ under independence and for $\hrhoamn$ under independence and vanishing third moments, then the distribution function of $\max (|U|,|V|)$ is $F(z)=(2\Phi(z)-1)^2$, $z>0$.
	%
	%
	
	%
	%
	%
\end{enumerate}
\end{theorem}	
We give the proof in the appendix. For convenience in the following Lemma \ref{lem:maxnormal} we give the distribution function and the density of $\max(U,V)$ arising in part \textit{(ii)} in the previous theorem. The density is also given in \cite{ker2001}, equation (1), while the (distinct) expressions for density and distribution function of $\min(U,V)$ in (46.77) and (46.77) of Sec. 46.6 in \cite{kotz2000}, or in \cite{cain1994} should also lead to the same results.  
\begin{lemma}\label{lem:maxnormal}
	Suppose that
	\begin{equation}\label{eq:bivariatenormal}
		(U,V) \sim \mathcal N\left(0,\begin{pmatrix} \sigma_1^2 & \tau \, \sigma_1 \, \sigma_2 \\ \tau \, \sigma_1 \, \sigma_2 & \sigma_2^2\end{pmatrix}\right).
	\end{equation}
	Then for $Z = \max(U,V)$ with density $f_Z$ we have that
	\begin{align}
		\PP(Z \leq z) & = \int_{- \infty}^z\, \varphi(t;0, \sigma_1^2)\, \Phi\Big(\frac{z - \tau \, \sigma_2 \, t / \sigma_1}{\sigma_2\, \sqrt{1-\tau^2}} \Big)\, \dd t,\label{eq:distrmax}\\
			f_Z(z) & =  \varphi(z;0, \sigma_1^2)\, \Phi\Big(\frac{z - \tau \, \sigma_2 \, z / \sigma_1}{\sigma_2\, \sqrt{1-\tau^2}} \Big)\, + \varphi(z;0, \sigma_2^2)\, \Phi\Big(\frac{z - \tau \, \sigma_1 \, z / \sigma_2}{\sigma_1\, \sqrt{1-\tau^2}} \Big),\label{eq:densermax}
	\end{align}	
	where $\varphi(\cdot;\mu, \sigma^2)$ is the density of the $ \mathcal N(\mu, \sigma^2)$ distribution.
	
	%
%
\end{lemma}
We give the calculations for the result in the Lemma in the supplementary material.
	For the special case $\sigma_1=\sigma_2=:\sigma$, the density is
	\begin{equation}
		g(z;\sigma,\alpha) = \frac{2}{\sigma}\varphi\Big(\frac{z}{\sigma}\Big)\, \Phi\Big(\alpha\, \frac{z}{\sigma} \Big), \quad \alpha=\sqrt{\frac{1-\tau}{1+\tau}}
	\end{equation}
	the density of a skew-normal distribution $SN(0,\sigma,\alpha).$ In general,
	\begin{equation} \label{sn-mix}
		f_Z(z) = \frac{1}{2} \, g\Big(z;\sigma_1, \frac{\sigma_1/\sigma_2-\tau}{\sqrt{1-\tau^2}} \Big) + \frac{1}{2} \, g\Big(z;\sigma_2, \frac{\sigma_2/\sigma_1-\tau}{\sqrt{1-\tau^2}} \Big),
	\end{equation}
	a mixture of two skew-normal distributions with equal weights.

\begin{remark}[Confidence intervals]\label{rem-ci}
	Constructing confidence intervals for $\rhoam$ and $\rhoamn$ is slightly complicated by the various cases for the asymptotic distribution as detailed in Theorem \ref{lem:asympdistrmax}. First, we need to estimate the asymptotic covariance matrix in \eqref{eq:asympcoeff}. For $\rhoamn$ we can substitute empirical moments into the expressions for $\Sigma^*$ as presented in the appendix. Estimation by substitution of the asymptotic covariance matrix of $\rhoam$ would be more complicated, therefore we prefer to use bootstrap estimates of the asymptotic covariance, which can also be used in the case of $\rhoamn$. 
    For details on this bootstrap procedure, see Section \ref{sec-asym-ci}.
	Then for constructing confidence intervals, one possibility is to ignore the cases $|\rho_1| = |\rho_2|$ in Theorem \ref{lem:asympdistrmax}, and simply work with the asymptotic distribution as specified in Theorem \ref{lem:asympdistrmax}, (i) depending on the sizes and signs of the estimates $\hrho_j$, $j=1,2$. In our simulations, these intervals were only slightly anti-conservative in case $|\rho_1| = |\rho_2|$.
	
	Alternatively, observing that the asymptotic distribution of $\max(U,V)$ or $\max(-U,V)$ in Lemma \ref{lem:maxnormal}  is stochastically larger than the normal distributions of $U$ and $V$, we may construct conservative intervals at level at least $1-\alpha$ by using as right end point the $1-\alpha/2$ - quantile of $\max(U,V)$ (respectively  $\max(-U,V)$)
  as well as the $\alpha/2$ - quantile of the normal distribution of either $U$ or $V$ (depending on the sizes of $|\hrho_1|$ and $ |\hrho_2|$). A detailed description of the different types of confidence intervals can be found in subsection  \ref{sec-asym-ci}.
\end{remark}

\begin{remark}[Testing for independence]
Both $\hrhoam$ and $\hrhoamn$ can be used for testing for independence of $(X,Y)$. Since maximum correlation vanishes if and only if $X$ and $Y$ are independent, such tests will be consistent in particular against alternatives for which $(\tilde X, \tilde Y)$ respectively $(\breve X, \breve Y)$ have a density of the form \eqref{eq:Lancasterexpansion}.

For $\hrhoam$, we can construct the critical value based on the quantile of the asymptotic distribution $F(z)=(2\Phi(z)-1)^2$, $z>0$, in Theorem \ref{lem:asympdistrmax}, (iii), alternatively we can use a permutation test. 
For details on the permutation procedure, see Section \ref{sec-ind-test}.
In our simulations, both methods perform very similarly.
As for $\hrhoamn$, 
apart from a permutation test, we can also use this asymptotic distribution function when assuming vanishing third moments. Alternatively, we can estimate the parameter $\tau$ in \eqref{eq:asymcovmat-ind} using empirical moments, and then work with the distribution of $\max(|U|,|V|)$ as specified in Lemma \ref{lem:maxnormal1} in the appendix. All three methods performed rather similarly in our simulations.
	
%
\end{remark}


\section{Numerical illustrations} \label{sec4}

\subsection{Asymptotic confidence intervals} \label{sec-asym-ci}

Here, we empirically study the coverage probabilities and mean length of six types of confidence intervals for the new coefficients, using the distributions of the previous subsection.

Case 1: First, let us consider $\rhoamn$, and assume $|\rho_1| \neq |\rho_2|$. Based on Theorem \ref{lem:asympdistrmax} and Remark \ref{rem-ci}, a confidence interval with asymptotic coverage probability $1-\alpha$ is given by
\begin{align} \label{ci-typ1}
    \left[ \max\{ \hrhoamn - z_{1-\alpha/2}\, s/\sqrt{n}, 0\},  \;
           \min\{ \hrhoamn + z_{1-\alpha/2}\, s/\sqrt{n}, 1\} \right],
\end{align}
where $z_{1-\alpha/2}=\Phi^{-1}(1-\alpha/2)$, and $s$ is an estimator of $\Sigma^*_{1,1}$ if $|\rho|>|\rho_2|$, and an estimator of $\Sigma^*_{2,2}$, else.
Here, $\Sigma^*_{k,k}$ are the diagonal elements of $\Sigma^*$ defined in  Proposition \ref{lem:asympnormest}.
To estimate $\Sigma^*$ we use plug-in, replacing theoretical by empirical moments. Since, for small sample sizes, empirical variances can take negative values, we fix a small $\delta$, e.g. $\delta=10^{-6}$. If $\Sigma^*_{k,k}=0$ for $k=1$ or $k=2$, we set the corresponding value to $\delta$, and, additionally, we set $\Sigma^*_{1,2}=0$.

Confidence intervals based on this plug-in estimator together with the bounds in (\ref{ci-typ1})  are denoted by $\text{ci}_{L,l}^{(p)}$.
An alternative method using a bootstrap estimator of $\Sigma^*$ is denoted by   $\text{ci}_{L,l}^{(b)}$.
Similarly, using a bootstrap estimate of the covariance matrix in the asymptotic distribution of the rank estimator yields intervals denoted by $\text{ci}_{L}^{(b)}$.
For the bootstrap procedure, we resample with replacement from the observations and write $\widehat{\rho}_{L,b}^\ast$ for the estimate of $\rho_L$ from the $bth$ bootstrap sample, $b=1,\ldots,B$. The bootstrap mean and the bootstrap covariance matrix are given by 
\begin{align*}
 \bar{\rho}_L^\ast &= \frac1B \sum_{b=1}^B \widehat{\rho}_{L,b}^\ast, \quad
 \widehat{\Sigma}_B = \frac{1}{B-1} \sum_{b=1}^B \left(\widehat{\rho}_{L,b}^\ast - \bar{\rho}_L^\ast\right) \left(\widehat{\rho}_{L,b}^\ast - \bar{\rho}_L^\ast\right)^\prime.
\end{align*}
Since $\widehat{\Sigma}_B$ estimates $\Cov(\widehat{\rho}_{L})$,  $n\widehat{\Sigma}_B$ is the bootstrap estimate of $\Sigma^*$.

Case 2: Next, assume $|\rho_1|=|\rho_2|=a>0$. Putting $\sigma_1=(\Sigma^*_{1,1})^{1/2}, \sigma_2=(\Sigma^*_{2,2})^{1/2}$ and $\tau=\Sigma^*_{1,2}/(\sigma_1\sigma_2)$, the asymptotic distribution of $\sqrt{n}(\hrhoamn-a)$ is given in (\ref{sn-mix}).
Let $q_p$ denote the $p$-quantile of this distribution, which is, by definition, larger in the usual stochastic order than the asymptotic normal distribution for case 1.
Denoting by $[l_2,u_2]$ the corresponding confidence interval for $\rhoamn$ with asymptotic coverage probability $1-\alpha$, it follows that $l_2$ is smaller than the lower bound $l_1$ in case 1; likewise, $u_2<u_1$ for the upper bounds.
Hence, an asymptotic conservative confidence interval for $\rhoamn$ is given by
\begin{align} \label{ci-typ2}
    \left[ \max\{\hrhoamn - q_{1-\alpha/2}/\sqrt{n}, 0\},  \;
           \min\{ \hrhoamn + z_{1-\alpha/2}\, s/\sqrt{n}, 1\} \right].
\end{align}
This method is referred to as $\text{ci}_{L,l}^{(p,c)}$ and $\text{ci}_{L,l}^{(b,c)}$, depending on the covariance estimator;
$\text{ci}_{L}^{(b,c)}$ indicates the corresponding intervals for $\rhoam$.


Tables \ref{ci-tab1} and \ref{ci-tab2} show the empirical coverage probability and simulated mean length of the various confidence intervals for sample size $200$ and nominal value $\alpha=0.05$, based on $10000$ replications. Additional results for sample sizes $50$ and $800$ can be found in the Supplementary Material, see Tables \ref{supp-tab3}-\ref{supp-tab6} in Section \ref{sec-supp7}.

\begin{table}[ht]
\centering
\begin{tabular}{lrrrrrr}
  \hline
distribution & $\text{ci}_{L,l}^{(p)}$ & $\text{ci}_{L,l}^{(p,c)}$ & $\text{ci}_{L,l}^{(b)}$ & $\text{ci}_{L,l}^{(b,c)}$ & $\text{ci}_{L}^{(b)}$ & $\text{ci}_{L}^{(b,c)}$ \\
  \hline
BVN(0) & 0.85 & 0.94 & 0.86 & 0.94 & 0.86 & 0.93 \\
  BVN(0.5) & 0.94 & 0.97 & 0.94 & 0.98 & 0.94 & 0.97 \\
  BVN(0.95) & 0.94 & 0.99 & 0.94 & 0.99 & 0.97 & 0.97 \\
  MN1 & 0.90 & 0.92 & 0.91 & 0.92 & 0.90 & 0.91 \\
  MN2 & 0.95 & 0.96 & 0.95 & 0.97 & 0.95 & 0.96 \\
  MN3 & 0.94 & 0.97 & 0.94 & 0.97 & 0.94 & 0.97 \\
  MN & 0.90 & 0.95 & 0.90 & 0.95 & 0.86 & 0.94 \\ \hline
  BVT5(0) & - & - & - & - & 0.89 & 0.90 \\
  BVT2(0) & - & - & - & - & 0.86 & 0.88 \\
  BVT1(0) & - & - & - & - & 0.89 & 0.90 \\
  BVT5(0.2) & - & - & - & - & 0.95 & 0.96 \\
  BVT2(0.2) & - & - & - & - & 0.87 & 0.88 \\
  BVT1(0.2) & - & - & - & - & 0.89 & 0.90 \\ \hline
  UnifDisc & 0.95 & 0.97 & 0.96 & 0.97 & 0.96 & 1.00 \\
  UnifRhomb & 0.96 & 0.97 & 0.96 & 0.98 & 0.95 & 1.00 \\
  UnifTriangle & 0.94 & 0.97 & 0.94 & 0.97 & 0.94 & 0.97 \\ \hline
  RegLin1 & 0.94 & 0.98 & 0.94 & 0.99 & 0.95 & 0.98 \\
  RegLin2 & 0.94 & 0.98 & 0.94 & 0.98 & 0.94 & 0.98 \\
  RegQuad1 & 0.95 & 0.98 & 0.96 & 0.98 & 0.93 & 0.98 \\
  RegQuad2 & 0.95 & 0.98 & 0.96 & 0.98 & 0.92 & 0.96 \\
  RegTrig1 & 0.95 & 0.97 & 0.95 & 0.97 & 0.94 & 0.96 \\
  RegTrig2 & 0.94 & 0.97 & 0.95 & 0.97 & 0.94 & 0.96 \\
   \hline
\end{tabular}
\caption{Empirical coverage probability, $n=200$, $B=10000$. \label{ci-tab1}}
\end{table}

\begin{table}[ht]
\centering
\begin{tabular}{lrrrrrr}
  \hline
distribution & $\text{ci}_{L,l}^{(p)}$ & $\text{ci}_{L,l}^{(p,c)}$ & $\text{ci}_{L,l}^{(b)}$ & $\text{ci}_{L,l}^{(b,c)}$ & $\text{ci}_{L}^{(b)}$ & $\text{ci}_{L}^{(b,c)}$ \\
  \hline
BVN(0) & 0.21 & 0.21 & 0.21 & 0.21 & 0.21 & 0.21 \\
  BVN(0.5) & 0.21 & 0.28 & 0.21 & 0.28 & 0.21 & 0.27 \\
  BVN(0.95) & 0.03 & 0.06 & 0.03 & 0.06 & 0.03 & 0.08 \\
  MN1 & 0.34 & 0.35 & 0.34 & 0.36 & 0.31 & 0.33 \\
  MN2 & 0.33 & 0.35 & 0.34 & 0.36 & 0.32 & 0.34 \\
  MN3 & 0.32 & 0.35 & 0.33 & 0.35 & 0.30 & 0.33 \\
  MN & 0.21 & 0.22 & 0.21 & 0.22 & 0.21 & 0.21 \\ \hline
  BVT5(0) & - & - & - & - & 0.31 & 0.32 \\
  BVT2(0) & - & - & - & - & 0.33 & 0.36 \\
  BVT1(0) & - & - & - & - & 0.26 & 0.33 \\
  BVT5(0.2) & - & - & - & - & 0.32 & 0.34 \\
  BVT2(0.2) & - & - & - & - & 0.33 & 0.36 \\
  BVT1(0.2) & - & - & - & - & 0.25 & 0.32 \\ \hline
  UnifDisc & 0.18 & 0.21 & 0.19 & 0.21 & 0.10 & 0.15 \\
  UnifRhomb & 0.12 & 0.15 & 0.13 & 0.15 & 0.07 & 0.12 \\
  UnifTriangle & 0.19 & 0.20 & 0.19 & 0.21 & 0.22 & 0.27 \\ \hline
  RegLin1 & 0.13 & 0.19 & 0.13 & 0.20 & 0.13 & 0.21 \\
  RegLin2 & 0.19 & 0.23 & 0.19 & 0.23 & 0.19 & 0.24 \\
  RegQuad1 & 0.28 & 0.32 & 0.28 & 0.32 & 0.16 & 0.24 \\
  RegQuad2 & 0.27 & 0.30 & 0.28 & 0.30 & 0.21 & 0.26 \\
  RegTrig1 & 0.22 & 0.25 & 0.22 & 0.25 & 0.20 & 0.21 \\
  RegTrig2 & 0.23 & 0.26 & 0.23 & 0.26 & 0.22 & 0.23 \\
   \hline
\end{tabular}
\caption{Simulated mean length of confidence intervals, $n=200$, $B=10000$. \label{ci-tab2}}
\end{table}

First, we note that the three intervals $\text{ci}_{L,l}^{(p)}$, $\text{ci}_{L,l}^{(b)}$ and $\text{ci}_{L}^{(b)}$ have similar empirical coverage; the same holds for their conservative counterparts.
The empirical coverage is close to the nominal value for most distributions for the anti-conservative intervals. This does not hold for BVN(0) and MN, i.e. under independence, where the conservative intervals work much better.
Due to the slow convergence of the estimators to the theoretical values in case of the $t$-distribution, coverage is exact only for very large sample sizes.

Concerning mean length, $\text{ci}_{L,l}^{(p)}$, $\text{ci}_{L,l}^{(b)}$ and $\text{ci}_{L}^{(b)}$ again behave quite similarly. The conservative counterparts  often nearly have  the same length; but there are also some cases where they are somewhat longer. Still, to obtain adequate coverage also for distributions under independence we generally recommend the conservative intervals.

\subsection{Testing for independence} \label{sec-ind-test}

In this subsection, we compare the empirical power of several tests for independence with our new proposals. Besides tests based on $\hrpear$ and $\hrspear$, conducted as permutation tests,
we choose the permutation test based on $\hrhoamn$. The results for the test based on the same statistic assuming independence and vanishing third moment (see Theorem \ref{lem:asympnormest}) as well as using the general asymptotic distribution under independence  are very similar and, therefore, omitted.
The permutation test works as follows. 
\begin{enumerate}
\item Compute $\hrhoamn=\hrhoamn((X_1,Y_1),\ldots,(X_n,Y_n)$.
\item For $b=1,\ldots,B$:
\begin{enumerate}
\item Simulate a randomly permuted sample $(X_1^{*b},\ldots,X_n^{*b})$ from $(X_1,\ldots,X_n)$.
\item Compute $\widehat{\rho}_{L,l}^{*b}=\hrhoamn((X_1^{*b},Y_1),\ldots,(X_n^{*b},Y_n)$.
\end{enumerate}
\item Compute the p-value approximation 
$$
 p^\ast = \frac{1}{B+1} \left(1 +\sum_{b=1}^B \ind\{\widehat{\rho}_{L,l}^{*b}>\hrhoamn\} \right).
$$
\end{enumerate}
For all permutation tests, we set $B=999$.

Tests based on $\hrhoam$ using the asymptotic distribution under independence and a permutation procedure 
are denoted by $\widehat{\rho_{\text{L(a)}}}$ and $\widehat{\rho_{\text{L(p)}}}$, respectively.
Further, we use permutation tests based on $\hdcor$ using the R function
\verb+hdcor.test+, on $\tstar$, as implemented in the R function \verb+tauStarTest+, and on $\hbcor$, using the R function \verb+bcov.test+. Next, a Monte Carlo permutation test based on $\hxi$ using function \verb+xicor+ is applied. Finally, we used the R function \verb+dhsic.test+ with a Gaussian kernel in the R package \verb+dHSIC+ \citep{dHSIC:2019} to perform a permutation test based on the Hilbert Schmidt independence criterion of \citet{zbMATH05070116}.


Table \ref{test-tab1} shows the results for sample sizes $100$, based on $10000$ replications. Additional results for sample sizes $25$ and $400$ can be found in Section 6 in the Supplementary Material.
In Table \ref{test-tab1}, the following set of distributions is utilized:

\begin{table}[ht]
\centering
\setlength{\tabcolsep}{4pt}
\begin{tabular}{lrrrrrrrrrrr}
  \hline
distribution & \hrpear & \hrspear & \hrhoamn & $\widehat{\rho_{\text{L(a)}}}$ & $\widehat{\rho_{\text{L(p)}}}$ & \hrace & \hdcor & \tstar & \hxi & \hhsic & \hbcor \\
  \hline
  BVN(0) & 0.05 & 0.05 & 0.05 & 0.05 & 0.05 & 0.05 & 0.05 & 0.05 & 0.05 & 0.05 & 0.05 \\ 
  BVN(0.5) & 1.00 & 1.00 & 1.00 & 1.00 & 1.00 & 0.92 & 1.00 & 1.00 & 0.69 & 0.96 & 0.97 \\ 
  BVN(0.95) & 1.00 & 1.00 & 1.00 & 1.00 & 1.00 & 1.00 & 1.00 & 1.00 & 1.00 & 1.00 & 1.00 \\ 
  MN1 & 0.10 & 0.07 & 0.55 & 0.52 & 0.52 & 0.31 & 0.10 & 0.06 & 0.07 & 0.17 & 0.23 \\ 
  MN2 & 0.41 & 0.37 & 0.65 & 0.63 & 0.63 & 0.34 & 0.42 & 0.35 & 0.11 & 0.32 & 0.41 \\ 
  MN3 & 0.68 & 0.66 & 0.78 & 0.77 & 0.77 & 0.40 & 0.69 & 0.63 & 0.18 & 0.50 & 0.57 \\ 
  MN & 0.05 & 0.05 & 0.05 & 0.05 & 0.05 & 0.05 & 0.06 & 0.06 & 0.05 & 0.05 & 0.06 \\ 
  BVT5(0) & 0.16 & 0.06 & 0.40 & 0.40 & 0.40 & 0.30 & 0.15 & 0.05 & 0.06 & 0.11 & 0.19 \\ 
  BVT2(0) & 0.49 & 0.08 & 0.75 & 0.92 & 0.92 & 0.75 & 0.78 & 0.09 & 0.10 & 0.56 & 0.80 \\ 
  BVT1(0) & 0.76 & 0.10 & 0.90 & 1.00 & 1.00 & 0.97 & 0.99 & 0.16 & 0.22 & 0.99 & 1.00 \\ 
  BVT5(0.2) & 0.52 & 0.47 & 0.61 & 0.64 & 0.64 & 0.38 & 0.56 & 0.43 & 0.12 & 0.32 & 0.42 \\ 
  BVT2(0.2) & 0.62 & 0.43 & 0.81 & 0.95 & 0.95 & 0.80 & 0.90 & 0.45 & 0.17 & 0.73 & 0.89 \\ 
  BVT1(0.2) & 0.79 & 0.40 & 0.92 & 1.00 & 1.00 & 0.97 & 0.99 & 0.52 & 0.30 & 0.99 & 1.00 \\ 
  UnifDisc & 0.02 & 0.02 & 0.92 & 0.89 & 0.89 & 0.30 & 0.05 & 0.05 & 0.08 & 0.27 & 0.18 \\ 
  UnifDrhomb & 0.00 & 0.01 & 1.00 & 1.00 & 1.00 & 0.91 & 0.13 & 0.10 & 0.15 & 0.94 & 0.91 \\ 
  UnifTriangle & 1.00 & 1.00 & 1.00 & 1.00 & 1.00 & 0.96 & 1.00 & 1.00 & 0.79 & 1.00 & 1.00 \\ 
  GARCH(2,1) & 0.21 & 0.08 & 0.78 & 0.82 & 0.82 & 0.63 & 0.43 & 0.10 & 0.10 & 0.53 & 0.67 \\ 
  RegLin1 & 1.00 & 1.00 & 1.00 & 1.00 & 1.00 & 1.00 & 1.00 & 1.00 & 1.00 & 1.00 & 1.00 \\ 
  RegLin2 & 1.00 & 1.00 & 1.00 & 1.00 & 1.00 & 0.98 & 1.00 & 1.00 & 0.83 & 1.00 & 1.00 \\ 
  RegQuad1 & 0.11 & 0.13 & 0.96 & 0.99 & 0.99 & 1.00 & 1.00 & 1.00 & 1.00 & 1.00 & 1.00 \\ 
  RegQuad2 & 0.08 & 0.10 & 0.55 & 0.64 & 0.64 & 1.00 & 1.00 & 1.00 & 1.00 & 1.00 & 1.00 \\ 
  RegTrig1 & 0.97 & 0.96 & 0.95 & 0.94 & 0.94 & 1.00 & 1.00 & 1.00 & 1.00 & 1.00 & 1.00 \\ 
  RegTrig2 & 0.87 & 0.87 & 0.81 & 0.82 & 0.82 & 1.00 & 1.00 & 1.00 & 1.00 & 0.98 & 1.00 \\ 
   \hline
\end{tabular}
\caption{Empirical power for several tests of independence with $n=100$, $B=10000$. \label{test-tab1}}
\end{table}

\smallskip\noindent
- Bivariate normal distributions with $\rpear=0,0.5,0.95$, called BVN($\rpear$). \\
- The normal mixture distributions NM1, NM2 defined in Section \ref{sec2}, and a third mixture with $p=1/4$, called NM3, for which $\rpear$ and $\rho_{max}$ coincide, taking the value 1/4. \\
- The normal mixture $0.25(\phi_{(0,0)}+\phi_{(0,5)}+\phi_{(5,0)}+\phi_{(5,5)})$, satisfying $H_0$, is called NM. Here, $\phi_{(\mu,\nu)}$ denotes the density of an uncorrelated bivariate normal distribution with expected value $(\mu,\nu)$ and standard deviation 1 for each marginal.  \\
- Bivariate $t$-distributions $\text{BVT}_{\nu}(\rho)$ with degree of freedom $\nu=5,2,1$ and linear correlation coefficient $\rho=0$ and $0.2$.
The linear correlation $\rho$ is the entry in the defining correlation matrix, and corresponds to $\rpear$ if the latter exists, i.e. for $\nu>2$.
Contrary to $\rpear$, Spearman correlation always exists. Its value is $\rspear=0$ for $\rho=0$, irrescpective of $\nu$. If $\rho=0.2$, $\rspear=0.186,0.179,0.169$ for $\nu=5,2,1$  \citep{heinen2020}. \\
- Besides UnifDisc from Section \ref{sec2}, we use a uniform distribution on the rhombus with vertices $(-1,0),(0,1),(0,1),(0,-1)$, called UnifRhomb, where $\rpear=0$ and $\rhom=0.5$. Further, a uniform distribution on the triangle with vertices $(0,0),(0,1),(1,0)$, called  UnifTriangle, where $\rhom=\rpear=0.5$ \citep{buja1990remarks}. \\
- GARCH(2,1), defined in Section \ref{sec2}. \\
- Finally, we use functional dependencies similar to Fig. 2 in \cite{chatterjee2021new}.
Data are generated as
\begin{align*}
X_i\sim U(a,b), \quad Y_i=f(X_i)+\epsilon_i, \text{ with } \epsilon_i\sim N(0,\sigma^2),
\end{align*}
for $i=1,\ldots,n,$ where $X_i$ and $\epsilon_i$ are independent. \\
Linear regression: $X_i\sim U(0,1), \ f(x)=x$. For RegLin1, we use $\sigma=0.3$, which results in $\rpear=0.69$. RegLin2 has $\sigma=0.45$ and $\rpear= 0.54$. \\
Quadratic regression: $X_i\sim U(-1,1), \ f(x)=x^2$. Here, $\rpear=\rspear=0$.
For RegQuad1 and RegQuad2, $\sigma=0.15$ and $\sigma=0.3$, respectively. \\
Trigonometric regression: $X_i\sim U(0,4\pi), \ f(x)=(\sin(x)+1)/2$.
For RegTrig1 and RegTrig2, $\sigma=0.15$ and $\sigma=0.3$.

\smallskip\noindent
The general finding is that all tests keep the theoretical level closely, even with small sample sizes. The tests based on $\hrhoamn$ and $\hrhoam$ have higher power against the alternatives MN1, MN2, MN3, BVT5(0), BVT5(0.2), UnifDisc, UnifRhomb and GARCH(2,1) than all competitors. For other multivariate $t$-distributions, the test based on $\hrhoam$ has the highest power, followed by $\hdcor$ and $\hbcor$, but $t^*$ and $\widehat{\xi}$ based tests have low power. Interestingly, the permutation test based on $\hrace$ has significantly less power than the test based on $\hrhoam$ for all these alternatives.
This behaviour remains similar even for the rather large sample size of 400, as can be seen in Table \ref{supp-tab2} in the supplement. 
In contrast, the newly proposed tests (as well as $\hrpear$ and $\hrspear$) have generally lower power than the other tests for the regression type alternatives RegQuad and RegTrig. For sample size $n=25$, while the power of all tests is rather low against several alternatives, the distance correlation-based test works as well as the tests based on our new correlation measures. Against the linear regression alternative RegLin, the $\hrace$ based test has the lowest power. Summing up, our new tests are sensitive against various types of deviations from independence, in the simulation scenarios considered even more so than distance correlation-based tests. In contrast, the tests based on $t^*$ and $\widehat{\xi}$ focus strongly on but seem limited to functional relationships.

%
\section{Data examples}\label{sec:dataex}
\subsection{Salaries for Professors} \label{sec-ex1}

\cite{fox2019}, Sec. 10.7.1, analyze the 2008-09 nine-month academic salary for Assistant Professors, Associate Professors and Professors in a college in the U.S., depending on several covariates.
The data are available as dataset \emph{Salaries} in the R package \emph{carData} \citep{fox2022}. Here, we consider the variables \emph{salaries} and \emph{year of service} in the subgroup of male professors with sample size $n=248$, grouped in discipline A ("theoretical", $n=123$) and discipline B ("applied", $n=125$). Adding the subgroup of female professors ($n=18$) yields nearly identical results. The subgroups of assistant and associate professors are clearly separated with only a few years of service.

Figure \ref{fig:ex1} shows scatterplots of salary vs. years of service for the original data, the data after applying the rank-based inverse normal (RIN) transformation, and for the squared data after RIN transformation.
The values of the various dependence measures used in Sec. \ref{sec2} are given in Table \ref{tab:ex1}.
The $p$-values of the different tests of dependence used in Section \ref{sec4} are shown in lines 2-4 of Table \ref{tab:ex1-2}; here,  the number of permutations was set to 9999 for all permutation tests.

\begin{figure}
    \centering
    \includegraphics[width=\textwidth]{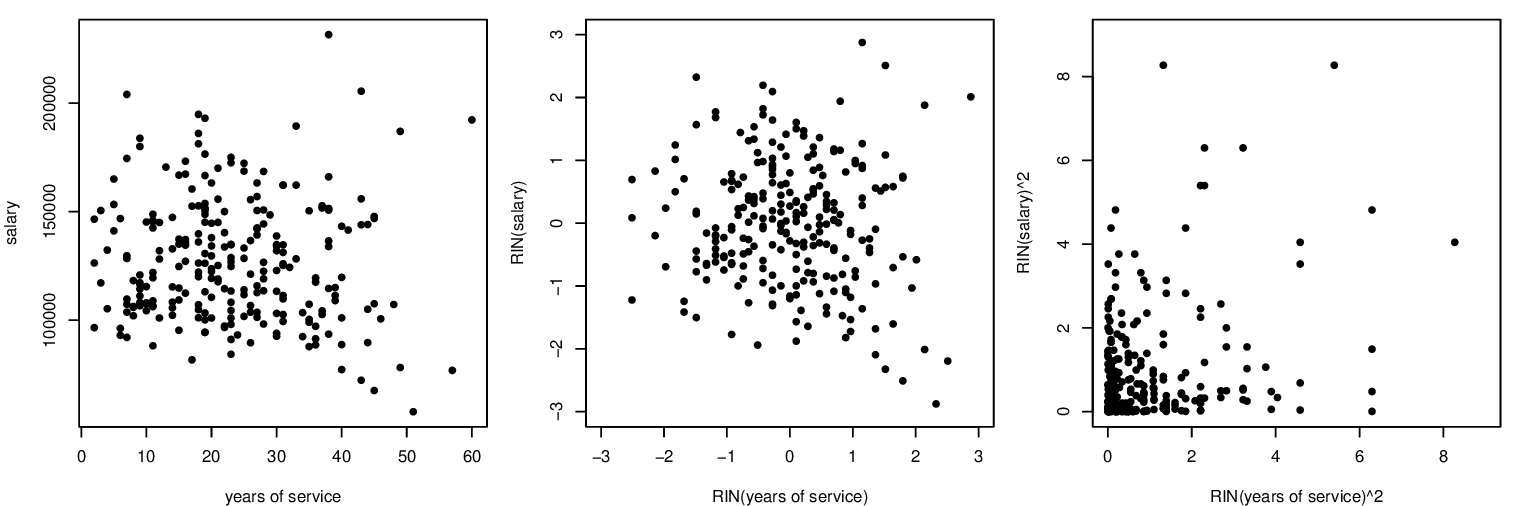}
    \caption{Scatterplots of salaries against years of service. Left: Original data. Middle: RIN transformed data. Right: RIN transformed and squared data.}
    \label{fig:ex1}
\end{figure}

\begin{table}
\centering
\begin{tabular}{c|ccccccccc}
& \hrpear & \hrspear & \hrhoamn & \hrhoam & \hrace & \hdcor & \tstar & \hxi & \hbcor \\ \hline
both         &  -0.07 &   -0.08 & 0.29 & 0.28  & 0.46  & 0.15 & 0.00 & 0.03 & 0.01 \\
discipline A & -0.19  &  -0.15  & 0.23 & 0.30  & 0.58  & 0.21 & 0.01 & 0.04 & 0.01 \\
discipline B &  0.13  &   0.06  & 0.29 & 0.21  & 0.43  & 0.19 & 0.00 &-0.02 & 0.01 \\
\end{tabular}
\caption{Values of dependence measures for datasets in Sec. \ref{sec-ex1}.}
\label{tab:ex1}
\end{table}

{\bfseries Findings:} The values of $\hrhoamn, \hrhoam$ and $\hrace$ are clearly different from zero, with $\hrace$ taking values around 0.5 for the combined data as well as for the two disciplines. 
The values of $\hrhoamn$ and $\hrhoam$ are between 0.2 and 0.3, with $p$-values close to zero for the combined data and for discipline A, and below 0.05 for discipline B. Thus, they are very similar to the p-values for the $\hrace$-based test.
Apart from the $p$-values 0.039 and 0.035 of $\hrpear$ and $\hbcor$ in subgroup discipline A and the p-value 0.022 of $\hbcor$ for the combined dataset, all remaining $p$-values are greater than 0.05.
In particular, $t^*,\hxi$ and $\hhsic$ do not indicate any kind of dependence.
Furthermore, the estimates for $\hbcor$ are practically 0, as is the case for $t^*$ and $\hxi$, which makes the test results for $\hbcor$ somewhat questionable.

\begin{table}
\centering
\setlength{\tabcolsep}{3pt}
\begin{tabular}{c|ccccccccccc}
& \hrpear & \hrspear & \hrhoamn & $\widehat{\rho_{\text{L(a)}}}$ & $\widehat{\rho_{\text{L(p)}}}$ & \hrace &\hdcor & \tstar & \hxi & \hhsic & \hbcor \\ \hline
both         & 0.269 & 0.215 & 0.001 & 0.000 & 0.000 & 0.000 & 0.057 & 0.097 & 0.254 & 0.073 & 0.022 \\
discipline A & 0.039 & 0.105 & 0.027 & 0.002 & 0.004 & 0.002 & 0.052 & 0.076 & 0.237 & 0.053 & 0.035 \\
discipline B & 0.145 & 0.485 & 0.010 & 0.042 & 0.042 & 0.024 & 0.145 & 0.391 & 0.718 & 0.491 &  0.204  \\ \hline\hline
BTC2         & 0.354 & 0.946 & 0.001 & 0.000 & 0.001 & 0.001 & 0.001 & 0.001 & 0.024 & 0.001 & 0.001 \\
BTC3         & 0.579 & 0.838 & 0.142 & 0.001 & 0.002 & 0.010 & 0.012 & 0.155 & 0.027 & 0.023 & 0.006 \\
\end{tabular}
\caption{p-values of tests of dependence for data sets in Sec. \ref{sec-ex1} and \ref{sec-ex2}.}
\label{tab:ex1-2}
\setlength{\tabcolsep}{12pt}
\end{table}

\subsection{Log returns of crypto currencies} \label{sec-ex2}

In this section, we consider log-returns of three cryptocurrencies, namely Bitcoin (BTC), Ethereum (ETH) and Ripple (XRP).
According to coinmarketcap.com, BTC and ETH are the cryptocurrencies with the highest and second-highest market capitalization; XRP is number seven on the list.
We use hourly observations from May 16, 2018, to October 27, 2021, with a sample size of 30264, which corresponds to 1261 days. All prices are closing values in U.S. dollars of the Bitstamp Exchange obtained from CryptoDataDownload under the URL \url{https://www.cryptodatadownload.com/data/bitstamp}. 
The datasets are named \verb+Bitstamp_XXXUSD_1h.csv+, where XXX is one of BTC, ETH, XRP.
Returns are estimated by taking logarithmic differences.
These datasets have been considered by \cite{holzmann2023} in a forecasting context. Here, however, we are interested in a possible dependence of the pairs $(r_t,r_{t-1})$ similar to the simulated GARCH dataset in Sec. \ref{sec2}.

Figure \ref{fig:ex2} shows scatterplots of the pairs $(r_t,r_{t-1})$ for the three cryptocurrencies; values of the various dependence measures are given in lines 2-4 of Table \ref{tab:ex2}. Note that $\hdcor$ and $\tstar$ could not be computed due to the rather large sample size of $30264$.

\begin{figure}
  \centering
  \includegraphics[width=0.9\textwidth]{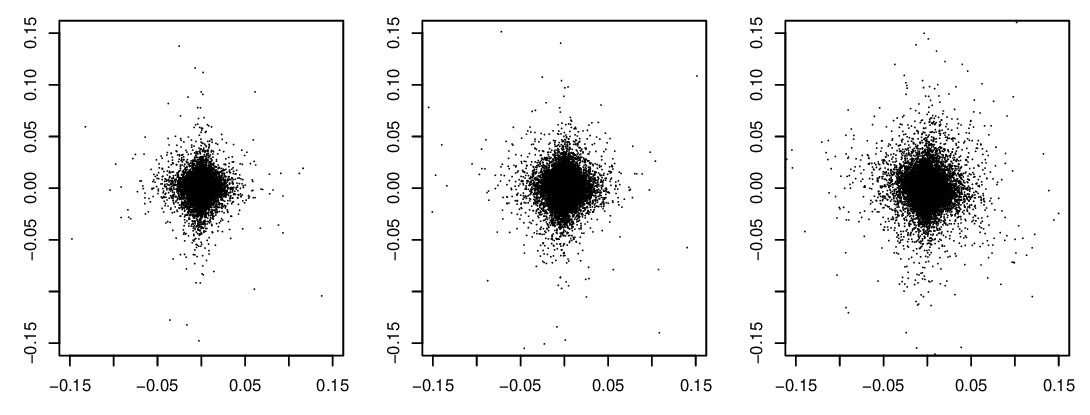}
  \caption{Scatterplots of the pairs of log returns $(r_t,r_{t-1})$ for the three cryptocurrencies. Left: Bitcoin (BTC), Middle: Ethereum (ETH), Right: Ripple (XRP).} \label{fig:ex2}
\end{figure}

\begin{table}
\centering
\begin{tabular}{c|ccccccccc}
 & \hrpear & \hrspear & \hrhoamn & \hrhoam & \hrace & \hdcor & \tstar & \hxi & \hbcor \\ \hline
BTC & -0.03 &  -0.07 &  0.22  & 0.29  & 0.29 &    - &    - & 0.01 & 0.01 \\
ETH & -0.01 &  -0.06 &  0.18  & 0.26  & 0.26 &    - &    - & 0.01 & 0.01 \\
XRP & -0.06 &  -0.08 &  0.18  & 0.36  & 0.35 &    - &    - & 0.02 & 0.01 \\ \hline
BTC2& -0.01 &   0.00 &  0.13  & 0.20  & 0.24 & 0.08 & 0.00 & 0.02 & 0.00 \\
BTC3& -0.02 &   0.01 &  0.06  & 0.13  & 0.23 & 0.11 & 0.00 & 0.04 & 0.00 \\ 
\end{tabular}
\caption{Values of dependence measures for dataset in Sec.~\ref{sec-ex2}}
\label{tab:ex2}
\end{table}


{\bfseries Findings:} As expected, the values of $\hrpear$ and $\hrspear$ are close to zero, and the same holds for $\hxi$ and $\hbcor$. The values of $\hrhoamn, \hrhoam$ and $\hrace$ are between 0.18 and 0.36, with $\hrhoam$ and $\hrace$ being nearly identical for all three datasets.
Since the datasets are heavy-tailed, the results for $\hrpear$ and $\hrhoamn$ should be judged with caution.

We also considered subsets of these datasets. As two examples, BTC2 and BTC3 contain the first six months and the first month of 2021 of the Bitcoin data with sample sizes of 4344 and 744, respectively.
The last two lines of Table \ref{tab:ex2} contain the values of the dependence measures for these datasets, which are similar to the results for the full dataset. Here,  $\hdcor$ has values comparable to $\hrhoamn$, and $\tstar$ takes values close to 0, as is the case for $\hxi$ and $\hbcor$. The  $p$-values of the corresponding tests are listed in the last two lines of Table \ref{tab:ex1-2}. 
For BTC2, all tests are significant at the 0.01 level except those based on $\hrpear, \hrspear$ and $\hxi$, the latter being significant at the 0.05 level.
For BTC3, only the $p$-values of $\widehat{\rho_{\text{L(a)}}}, \widehat{\rho_{\text{L(p)}}}, \hrace$ and $\hbcor$ are significant at the 0.01 level. Again, the test results for $\hbcor$ are difficult to interpret since $\hbcor$ is close to 0.
%
%



\section{Concluding remarks}\label{sec:conclude}

One current major strand in the literature on correlation coefficients focuses on defining notions of correlation for random variables with values in possible distinct high-dimensional \citep{szekely2007measuring, Zhu:2021} or infinite-dimensional spaces \citep{pan:2020} or even abstract metric spaces \citep{nies2021transport}. Another strand of the literature considers correlation coefficients designed to detect regression-type dependence \citep{dette2013copula, chatterjee2021new}. 

Correlation coefficients may serve as mere tools for independence tests, say in combination with permutation methods. However, when the actual estimated values of these coefficients are supposed to have some appeal for practitioners, in our assessment one should not neglect their performance in the most classical situation, that of a bivariate normal. Our coefficients and the general maximum correlation equal the Pearson correlation for the normal distribution, and the distance correlation has slightly distinct but reasonably similar values. However, as we demonstrate in Section \ref{sec2} some recently proposed correlation coefficients such as the $\tau^*$ coefficient of \cite{bergsma2014consistent}, the $\xi$-coefficient of \cite{chatterjee2021new} or the ball correlation coefficient of \cite{pan:2020} may have very low values for correlated Gaussians in the order $1/10$ or even $1/100$ of that of the actual Pearson correlation, which makes their practical use beyond independence tests somewhat problematic. 
Except for regression-type dependence, our measures have higher power against various alternatives than several competing methods, including maximum correlation when using the ACE algorithm.  

For two jointly normally distributed random vectors $X$ and $Y$ of possibly distinct dimensions, maximum correlation equals the maximal canonical correlation \citep{zbMATH03173142}
$$ \rho_{\text{can}}(X,Y) = \max_{a,b}\rpear(a^\top X, b^\top Y) = \rpear(u^\top X, v^\top Y)$$
where $a$ and $b$ vary through the Euclidean spaces of appropriate dimensions, and $(u^\top X, v^\top Y)$ is the first pair of canonical variables. Various possibilities to extend the Lancaster correlation coefficient to multiple dimensions now present themselves, the simplest being by applying the univariate Lancaster coefficients to the canonical variables resulting in $\rhoam(u^\top X,v^\top Y)$ or in  $\rhoamn(u^\top X,v^\top Y)$. A further investigation of multivariate extensions of our methods is intended in a subsequent work. 


\section*{Supplementary material} \label{SM}
Supplementary material  
includes
the proof of Proposition \ref{lem:asympnormest} and Lemmas \ref{lem:maxnormal} and \ref{lem:maxnormal1} (below) along with additional simulation results.


\part*{Appendix}
\appendix

\section*{Asymptotic covariance matrix in Proposition \ref{lem:asympnormest}}




\begin{sloppypar}

We specify the components $\sigma_m$ and $M$ in the representation $\Sigma^* = M \, \Sigma_m \, M^\top $ of the asymptotic covariance matrix in Proposition \ref{lem:asympnormest}. The derivation using the $\Delta$ - method is deferred to the appendix.  

$\boldsymbol{\Sigma_m}$: Recall the notation $\breve X = (X - \Ex[X])/\text{sd}(X)$ and similarly for $\breve Y$. 
Let 
$$c_{kl,sr} = \Cov(\breve X^k \breve Y^l, \breve X^s \breve Y^r).$$ 
Then $\Sigma_m = (c_{kl,sr})$ for $kl$ and $sr$ ranging through  $\{ 10,01,20,02,11,30,03,$ $21,12,40,04,22\}$.
\end{sloppypar}

$\boldsymbol{M}$: We shall represent $M = B\, A$ and specify the matrices $A$ and $B$. Set 
$$e_{kl} = \Ex[\breve X^k \breve Y^l],$$ then
%
\begin{align*}
	A & = \left(\begin{array}{cccccccccccc}
	0 & 0 & 1 & 0 & 0 & 0 & 0 & 0 & 0 & 0 & 0 & 0 \\
	0 & 0 & 0 & 1 &  0 & 0 & 0 & 0 & 0 & 0 & 0 & 0 \\
	0 & 0 & 0 & 0 &  1 & 0 & 0 & 0 & 0 & 0 & 0 & 0 \\
	- 4 e_{30}&  0 & - 2 e_{40} & 0 & 0 & 0 & 0 & 0 & 0 & 1 & 0 & 0 \\
	0 &  - 4 e_{03}&  0 & - 2 e_{04} & 0 & 0 & 0 & 0 & 0 & 0 & 1  & 0 \\
	- 2 e_{12} &  - 2 e_{21}&  -  e_{22} & - e_{22} & 0 & 0 & 0 & 0 & 0 & 0 & 0  & 1 \\
\end{array}\right),\\
B & = 	\left(\begin{array}{cccccc}
	- \rhono/2 & - \rhono/2 & 1 & 0 & 0 & 0  \\
	0 & 0 & 0 & - \frac{\rhont}{2\, (e_{40}-1)}  &  - \frac{\rhont}{2\, (e_{04}-1)}  & \frac{1}{\big( (e_{40} - 1)\,(e_{04} - 1)\big)^{1/2}}    \\
\end{array}\right)
\end{align*}
%
Explicitly one obtains
\begin{align*}
	\Sigma^*_{1,1} &= (e_{40}+e_{04}+2e_{22})\rho^2/4 - (e_{31}+e_{13})\rho+e_{22},
\end{align*}
but rather complicated expressions for $\Sigma^*_{1,2}$ and $\Sigma^*_{2,2}$.

\section*{Proof of Proposition \ref{lem:jointasympnormalrankgeneral}}

\begin{proof}[Proof of Proposition \ref{lem:jointasympnormalrankgeneral}]

\textit{Asymptotic distribution under independence}

Let $R^\circ(j)$ denote the anti-ranks of the $Y_j$ given sorted $X_j$'s \citep[Section 13.3]{van2000asymptotic}, so that
\begin{align}\label{eq:defsn}
	S_n : = n\,  s_a^2 \, \hrhoro = \sum_{j=1}^n a(Q(j))a(R(j))
	&= \sum_{j=1}^n a(j) a(R^\circ(j)).
\end{align}
If  $X$ and $Y$ are independent, $S_n$ is distributed as the linear rank statistic
$$S_n' = \sum_{j=1}^n a(j) a(R(j)).$$
For the scores, we have that
$$s_a^2 = \frac1n\, \sum_{j=1}^n \, \Big(\Phi^{-1}\big(\frac{j}{n+1}\big)\Big)^2 \to \int_0^{1} (\Phi^{-1}(u))^2 \dd u = 1,$$
and furthermore $\Var(S_n')=\frac{n^2}{n-1}s_a^4$. Therefore, under independence from \citet[Theorem 13.5]{van2000asymptotic}
\begin{equation}\label{eq:expandoneind}
	\sqrt n \, \hrhoro \stackrel{\cL}{=} \frac{1}{\sqrt n\, s_a^2}\, S_n' = \frac{1}{\sqrt n\, s_a^2}\, \sum_{j=1}^n \, a(j)\, \Phi^{-1}(U_j) + o_\PP(1),
\end{equation}
where $U_j = F_Y(Y_j)$ are independent and uniformly distributed.

Similarly,
\begin{align}\label{eq:deftn}
	T_n : = n\,  s_b^2 \, \hrhort = \sum_{j=1}^n b(Q(j))b(R(j)) - n\  \bar b^2 &= \sum_{j=1}^n b(j) b(R^\circ(j)) - n\  \bar b^2,
\end{align}
which, if  $X$ and $Y$ are independent, is distributed as the linear rank statistic
$$T_n' = \sum_{j=1}^n b(j) b(R(j)) - n\  \bar b^2.$$
For the scores, we have that
$$\frac1n\, \sum_{j=1}^n \, \Big(\Phi^{-1}\big(\frac{j}{n+1}\big)\Big)^4 \to \int_0^{1} (\Phi^{-1}(u))^4 \dd u = 3,$$
so that
$ s_b^2 \to 2.$
Since $\Var(T_n')=\frac{n^2}{n-1}s_b^4$, we obtain that \citep[Theorem 13.5]{van2000asymptotic}
\begin{equation}\label{eq:expandtwoind}
	\sqrt n \, \hrhort \stackrel{\cL}{=} \frac{1}{\sqrt n\, s_b^2}\, T_n' = \frac{1}{\sqrt n\, s_b^2}\, \sum_{j=1}^n \, \big( b(j) - \bar b\big)\, \big(\Phi^{-1}(U_j)\big)^2 + o_\PP(1),
\end{equation}
with $U_j$ as defined above. To conclude with the joint asymptotic normal distribution in \eqref{eq:expandoneind} and \eqref{eq:expandtwoind}, first note that the variance of the linear part in \eqref{eq:expandoneind} is $1$, while in \eqref{eq:expandtwoind} it converges to $1$, and the covariance is $0$ since $\Ex[(\Phi^{-1}(U_j))^3] = 0$. For joint asymptotic normality, one may refer to the next part of the proof (without assuming independence). In a direct argument, using the Cramér-Wold device one needs to check the Lindeberg condition for the variables
$$ \lambda \, a(j)\, \Phi^{-1}(U_j) + \mu \, \big( b(j) - \bar b\big)\, \big(\Phi^{-1}(U_j)\big)^2,\qquad \lambda, \mu \in \R, j=1, 2, \ldots,$$
where $\mu\not=0$ or $\lambda \not=0$.  This is implied by the Noether condition
\begin{align*}
	&\frac{\lambda^2\, \max_{j=1}^n a^2(j) + \mu^2\,  \max_{j=1}^n b^2(j)}{\lambda^2\, \sum_{j=1}^n a^2(j) + \mu^2\, \sum_{j=1}^n \big( b(j) - \bar b\big)^2} \\
	&  \hspace*{2cm} = \frac{\lambda^2\,\big(\Phi^{-1}(n/(n+1))\big)^2 + \mu^2 \,\big(\Phi^{-1}(n/(n+1))\big)^2}{\lambda^2\, n\, s_a^2 + \mu^2\, n\, s_b^2} \ \to \ 0.
\end{align*}
This follows from $s_a^2 \to 1$, $s_b^2 \to 2$ and from the quantile bound
%

%
%
%
\begin{equation}\label{eq:logboundquantile}
	\Phi^{-1}(n/(n+1)) \leq \sqrt{2\, \log((n+1)/2)},
\end{equation}
which follows from the bound $1-\Phi(z)\leq \exp(-z^2/2)/2$, $z \geq 0$.
%
%

%
%
%

\vspace{0.5cm}

\textit{Asymptotic distribution in general}

%
We use \citet[Proof of Theorem 2.1]{ruymgaart1972asymptotic} to derive the joint asymptotic distribution of $(\hrhoro, \hrhort)$ in the general setting. 	

Let $C(u,v)$ denote the copula of $(X,Y)$ and observe that
$$\rhoro = \int_{(0,1)^2} \Phi^{-1}(u)\, \Phi^{-1}(v) \,  \dd C(u,v).$$			
%
%
After checking \citet[Assumption 2.1 and 2.2]{ruymgaart1972asymptotic} (see below), setting
$$(U_i,V_i) = (F_X(X_i),F_Y(Y_i)), \qquad i=1, \ldots, n,$$
from \citet[(3.9) and p.~1126]{ruymgaart1972asymptotic} we obtain the following asymptotic expansion for $S_n$ in \eqref{eq:defsn}:
\begin{align*}
	\sqrt n \, \big(\frac1n S_n - \rhoro \big) & = \frac{1}{\sqrt n}\, \sum_{i=1}^n \Big(\Phi^{-1}(U_i)\, \Phi^{-1}(V_i) - \rhoro \nonumber \\
	& \qquad \quad + \int_{(0,1)^2} \big(1([U_i, \infty))(u) - u\big) (\Phi^{-1})'(u)\, \Phi^{-1}(v) \, \dd C(u,v) \nonumber \\
	& \qquad \quad + \int_{(0,1)^2} \big(1([V_i, \infty))(v) - v\big) \, \Phi^{-1}(u) \,(\Phi^{-1})'(v) \,  \dd C(u,v) \Big)  + \, o_\PP(1) \nonumber \\
	& =  A_{n,1} + \, o_\PP(1),
\end{align*}		
and therefore
\begin{align}\label{eq:rank1asymp}
 \sqrt n \big(\hrhoro - \rhoro \big) & = \frac{\sqrt n}{s_a^2} \, \big(\frac1n S_n - \rhoro \big) \, + \rhoro \, \sqrt n \,  \frac{1-s_a^2 }{s_a^2} \nonumber\\
	& = \frac{1}{s_a^2}\, A_{n,1} + \, o_\PP(1),
\end{align}
by using Lemma \ref{lem:riemansumapprox} below on the second term.
Concerning 	\citet[Assumption 2.1 and 2.2]{ruymgaart1972asymptotic}, for $\Phi^{-1}$ this is immediate from the logarithmic bound (see \eqref{eq:logboundquantile}) 
\begin{equation}\label{eq:boundquantile}
	\Phi^{-1}(u) \leq \Big(2\, \log\big(1/(2(1-u)) \big) \Big)^{1/2}, \qquad u \in (1/2,1).
\end{equation}


For the derivative
$ (\Phi^{-1})'(u) = \frac{1}{\varphi(\Phi^{-1}(u))}$
we use the bound
\begin{equation}\label{eq:boundderivative}
	\varphi(\Phi^{-1}(u)) \geq \frac{\sqrt 2}{\sqrt \pi} \min(u, (1-u)), \qquad u \in (0,1).
\end{equation}
By symmetry, it suffices to check this for $u \in (0,1/2]$. Both sides equal $1/\sqrt{2 \pi}$ at $1/2$ and tend to $0$ as $u \to 0$. Therefore, it suffices to show that the left side is concave. But for the second derivative,
$$ \big(\frac{\dd^2}{\dd u^2} \varphi\circ \Phi^{-1}\big)(u) = \big(\frac{\dd^2}{\dd u^2} \log \varphi \big) (\Phi^{-1}(u)),$$
and concavity follows from the log-concavity of the normal density.


Similarly, denoting $\Phi^{-2}(u) : = (\Phi^{-1}(u))^2$ we have that
$$\rhort   = \big(\int_{(0,1)^2} \Phi^{-2}(u)\, \Phi^{-2}(v) \,  \dd C(u,v) - 1\big)/2 .$$			

After checking \citet[Assumption 2.1 and 2.2]{ruymgaart1972asymptotic}, now for $\Phi^{-2}(u)$, see below, for $T_n$ in \eqref{eq:deftn} we obtain the asymptotic expansion
%
%
\begin{align}\label{eq:rank1asymp1}
	& \sqrt n \, \big(\frac1n T_n + \bar b^2 - (2\, \rhort + 1) \big)\nonumber\\
	  = & \frac{1}{\sqrt n}\, \sum_{i=1}^n \Big(\Phi^{-2}(U_i)\, \Phi^{-2}(V_i) - (2\, \rhort + 1) \nonumber \\
	& \qquad \qquad + \int_{(0,1)^2} \big(1([U_i, \infty))(u) - u\big) (\Phi^{-2})'(u)\, \Phi^{-2}(v) \, \dd C(u,v) \nonumber \\
	& \qquad \qquad + \int_{(0,1)^2} \big(1([V_i, \infty))(v) - v\big) \, \Phi^{-2}(u) \,(\Phi^{-2})'(v) \,  \dd C(u,v) \Big)  + \, o_\PP(1) \nonumber \\
	 =  & A_{n,2} + \, o_\PP(1).
\end{align}		
Setting
$$\rranksn = \Big( \int_{(0,1)^2} \Phi^{-2}(u)\, \Phi^{-2}(v) \,  \dd C(u,v) - \bar b^2 \Big)/s_b^2,$$			
we obtain
\begin{align}
 \sqrt n \big(\hrhort - \rhort \big) & = \sqrt n \big(\hrhort - \rranksn \big) + \sqrt n \big(\rranksn - \rhort \big) \nonumber\\
 		& = 	\frac{1}{s_b^2}\, A_{n,2} + \, o_\PP(1) + o(1),\label{eq:saympexprankgen2}
\end{align}
by using Lemma \ref{lem:riemansumapprox}.
Asymptotic normality follows from \eqref{eq:rank1asymp} and \eqref{eq:saympexprankgen2} and the multivariate Lindeberg-L\'evy theorem.
%
%
%

%
%
%

Concerning 	\citet[Assumption 2.1 and 2.2]{ruymgaart1972asymptotic} for $\Phi^{-2}$, from \eqref{eq:boundquantile} we obtain

\begin{equation*}
	\Phi^{-2}(u) \leq \, 2\, \log\big(1/(2(1-u)) \big) , \qquad u \in (1/2,1).
\end{equation*}
Further, since
$$ \big(\Phi^{-2} \big)'(u)	= 2\,  \varphi(\Phi^{-1}(u))\, \Phi^{-1}(u)	$$
we can combine  \eqref{eq:boundderivative} and \eqref{eq:boundquantile} which takes care of the required assumptions.
\end{proof}
\begin{lemma}\label{lem:riemansumapprox}
	We have the following error bounds in the Riemann sum approximations:
	\begin{equation*}
		s_a^2 - 1 = \bar b - 1 = O\big((\log n)^{3/2} /n \big),\qquad   s_b^2 - 2 = O\big((\log n)^{5/2} /n \big).
	\end{equation*}
\end{lemma}
\begin{proof}
	We give the proof of the first claim, the second follows similarly. We estimate
	\begin{align}
		|s_a^2 - 1| & \leq 2\, \Phi^{-1}((n-1)/n) \,\frac1n\,  \sum_{j=2}^{n-1} \int_{(j-1)/n}^{j/n}\, |\Phi^{-1}(u) - \Phi^{-1}(j/(n+1))|\, \dd u \label{eq:sumbound}\\
			& + \frac1n\, \Big(2\, \Phi^{-2}(n/(n+1)) + 2\, \int_{(n-1)/n}^1 \Phi^{-2}(u)\, \dd u\Big).\nonumber
	\end{align}
	The first term in the bracket in the second line is bounded order $\log n$ by \eqref{eq:logboundquantile}, and the factor in the first line by $\sqrt{ \log n}$. Furthermore, for a standard normally distributed random variable $N$, setting $a = \Phi^{-1}((n-1)/n)$,
	\begin{align*}
		\int_{(n-1)/n}^1 \Phi^{-2}(u)\, \dd u & = \Ex\big[N^2\, 1(N \geq \Phi^{-1}((n-1)/n) ) \big] \\
			& = a\, \varphi(a) + (1 - \Phi(a)) \leq C\, \big(\sqrt{ \log n} / n + 1/n \big),
	\end{align*}
	since the normal quantile is also lower bounded by
	$$ a = \Phi^{-1}((n-1)/n) \geq \sqrt{2\, \log((n+1)/c)}$$
	for some $c>0$ and $n$ sufficiently large. Finally, for the sum in \eqref{eq:sumbound}, using the mean value theorem and \eqref{eq:boundderivative} for the derivative of $\Phi^{-1}$ we obtain a bound of order
	$$ \sum_{j=2}^{n-1}\, \frac1n\, \frac{n}{\max(j,n-j)} = O(\log n). $$
	Combining the estimates yields the first statement.
\end{proof}


\section*{Proof of Theorem \ref{lem:asympdistrmax}}

\begin{proof} 
\begin{itemize}
\item
\begin{sloppypar}
Consider $|\rho_2| < |\rho_1|$, the other case being similar. If $\rho_1>0$ then $\max( |\rho_1|,  |\rho_2|) = \rho_1 (>0)$. By consistency $\hat \rho_1 \to \rho_1$, so that $\PP(\hat \rho_1 = \max(|\hat \rho_1|, |\hat \rho_2|)) \to 1$. Therefore \eqref{eq:finalasymptotic} has the same asymptotic normal limit $U$ as $\sqrt n \, (\hat \rho_1 - \rho_1)$. If $\rho_1<0$ then $\max(|\rho_1|, |\rho_2|) = - \rho_1 (>0)$. Again by consistency $\hat \rho_1 \to \rho_1$, we have $\PP(- \hat \rho_1 = \max(|\hat \rho|, |\hat \rho_2|)) \to 1$, and \eqref{eq:finalasymptotic} has the same asymptotic distribution $-U$ as $\sqrt n \, ( - \hat \rho + \rho)$, but $U \stackrel{\cL}{=} -U$.
\end{sloppypar}
\item
Suppose first that  $\rho_1 = \rho_2 =:\rho >0$. By consistency, $\PP(\hat \rho_1 >0, \hat \rho_2 >0) \to 1$ so that
	$ \sqrt n \big(\max(|\hat \rho_1|, |\hat \rho_2|) - \rho\big)$ and $\max\big(\sqrt n (\hat \rho_1 - \rho), \sqrt n (\hat \rho_2 - \rho) \big) $
	are asymptotically equivalent with asymptotic distribution
	$ \max (U,V)$ by the continuous mapping theorem. Similarly, we deduce the asymptotic distribution of $ \max (-U,-V)$ in case $\rho_1 = \rho_2 <0$, which is the same as that of $ \max (U,V)$ since $(U,V) \stackrel{\cL}{=} (-U,-V)$.
	If $- \rho_1 = \rho_2 =: \rho>0$ then $\PP(\hat \rho_1 <0, \hat \rho_| >0) \to 1$, and
	$ \sqrt n \big(\max(|\hat \rho|, |\hat \rho_2|) - \rho\big)$ and $\max\big(\sqrt n (- \hat \rho -  ( - \rho)), \sqrt n (\hat \rho_2 - a) \big) $ are asymptotically equivalent, with asymptotic distribution  $ \max (-U,V)$. Finally, for the case $\rho_1 = -\rho_2 >0$ we obtain  $ \max (U,-V)$, which has the same distribution as  $ \max (-U,V)$ since $(-U,V) \stackrel{\cL}{=} (U,-V)$.
	%
	
		%
	\item If $\rho_1 = \rho_2 =0$, the weak convergence of 	$ \sqrt n \max(|\hat \rho_1|, |\hat \rho_2|) \big)$ to $\max (|U|,|V|)$ is immediate from the continuous mapping theorem.
\end{itemize}

\end{proof}	

\begin{lemma}\label{lem:maxnormal1}
	Again suppose that (U,V) is distributed according to (12). 
	%
	%
	Then for $Z = \max(|U|,|V|)$ we have for $z >0$ that
	\begin{equation}\label{eq:distrmax1}
		\PP(Z \leq z) = 2\, \int_0^z\, \varphi(t;0, \sigma_1^2)\,\Big( \Phi\Big(\frac{z - \tau \, \sigma_2 \, t / \sigma_1}{\sigma_2\, \sqrt{1-\tau^2}} \Big) - \Phi\Big(\frac{-z- \tau \, \sigma_2 \, t / \sigma_1}{\sigma_2\, \sqrt{1-\tau^2}} \Big)\,\Big)\, \dd t.
	\end{equation}	
	
	The density function $f_Z$ of $Z$ is given for $z>0$ by
	\begin{align}
		f_Z(z) & =  2\, \varphi(z;0, \sigma_1^2)\,\Big( \Phi\Big(\frac{z - \tau \, \sigma_2 \, z / \sigma_1}{\sigma_2\, \sqrt{1-\tau^2}} \Big)\, - \Phi\Big(\frac{-z - \tau \, \sigma_2 \, z / \sigma_1}{\sigma_2\, \sqrt{1-\tau^2}} \Big) \Big) \nonumber\\
		& + 2\, \varphi(z;0, \sigma_2^2)\, \Big(\Phi\Big(\frac{z - \tau \, \sigma_1 \, z / \sigma_2}{\sigma_1\, \sqrt{1-\tau^2}} \Big) - \Phi\Big(\frac{-z - \tau \, \sigma_1 \, z / \sigma_2}{\sigma_1\, \sqrt{1-\tau^2}} \Big) \Big)  \label{eq:densermax1}
	\end{align}
\end{lemma}
	Again, $f_Z$ can be expressed using densities of skew-normal distributions as
	\begin{align*}
		f_Z(z) = &g\Big(z;\sigma_1, \frac{\sigma_1/\sigma_2-\tau}{\sqrt{1-\tau^2}} \Big)
		- g\Big(-z;\sigma_1, \frac{\sigma_1/\sigma_2+\tau}{\sqrt{1-\tau^2}} \Big) \\
		&+ g\Big(z;\sigma_2, \frac{\sigma_2/\sigma_1-\tau}{\sqrt{1-\tau^2}} \Big)
		- g\Big(-z;\sigma_2, \frac{\sigma_2/\sigma_1+\tau}{\sqrt{1-\tau^2}} \Big)
	\end{align*}
	for $z>0$. Hence, the distribution function of $Z$ can be readily computed using the cdf of $SN(0,\sigma,\alpha).$

\bibliographystyle{chicago}
\bibliography{biblio}

\newpage		
\part*{Supplementary material}
\setcounter{page}{1}

\section{Covariance matrix \texorpdfstring{$\mathbf \Sigma^*$}{Sigma*} in (\ref{eq:asymcovmat-bvn})
under normality}

\begin{remark} 
If $(X,Y)$ are jointly normally distributed, for the quantities $\Sigma_m$, $A$ and $B$ given in the appendix to determine $\Sigma^*$ we observe that
	$e_{kl}=0$ if $k+l$ is odd and furthermore,
\begin{align*}
	e_{31} &= 3\rho, \quad e_{51} = 15\rho, \quad
	e_{22} = 2\rho^2+1, \quad e_{42} = 3(4\rho^2+1), \\
	e_{62} &= 15(6\rho^2+1),  \quad
	e_{33} = 3\rho(2\rho^2+3), \quad
	e_{44} = 3(8\rho^4+24\rho^2+3),
\end{align*}
\citep{kotz2000}. The nonzero entries in $\Sigma$ are
\begin{align*}
    c_{10,10} &= 1, \quad c_{10,01} = \rho, \quad
	c_{10,30} = 3, \quad c_{10,03} = 3\rho, \quad
	c_{10,21} = 3\rho, \quad c_{10,12} = 2\rho^2+1, \\
	c_{20,20} &= 2, \quad c_{20,02} = 2\rho^2, \quad
	c_{20,11} = 2\rho, \quad c_{20,40} = 12, \quad 
    c_{20,22} = 10\rho^2+2, \\
	c_{11,20} &= 2\rho, \quad c_{11,11} = \rho^2+1, \quad
	c_{11,40} = 12\rho, \quad c_{11,22} = 4\rho(\rho^2+2), \\
	c_{30,10} &= 3, \quad c_{30,01} = 3\rho, \quad
	c_{30,30} = 15, \quad c_{30,03} = 3\rho(2\rho^2+3), \quad
	c_{30,21} = 15\rho, \\  
	c_{21,10} &= 3\rho, \quad c_{21,30} = 15\rho, \quad
    c_{21,03} = c_{21,21} = 3(4\rho^2+1), \quad c_{21,12} = 3\rho(2\rho^2+3), \\
	c_{40,20} &= 12, \quad c_{40,02} = 12\rho^2, \quad
	c_{40,11} = 12\rho, \quad c_{40,40} = 96, \quad
	c_{40,04} = 24\rho^2(\rho^2+3), \\ 
	c_{22,11} & = 4\rho(\rho^2+2), \quad c_{22,40} = 12(7\rho^2+1), \quad
    c_{22,22} = 20\rho^4+68\rho^2+8,
\end{align*}
together with the corresponding nonzero entries when the roles of $U$ and $V$ are reversed.
Plugging these values into $\Sigma_m$ and computing $\Sigma^* = M \, \Sigma_m \, M^\top $,
yields the matrix in (\ref{eq:asymcovmat-bvn}).
\end{remark}

\section{Proof of Proposition 1} 

\begin{proof} 
%
%
Denote the empirical moments by
\begin{equation*}
	m_{kl} = \frac1n\, \sum_{i=1}^n X_i^k\, Y_i^l,\qquad k,l \in \{0,1,\ldots, 4\}.
\end{equation*}
Further, set
\begin{equation*}
	s_{xy} = m_{11} - m_{10} m_{01},\qquad s_x^2 = m_{20} - m_{10}^2,\qquad  s_y^2 = m_{02} - m_{01}^2,
\end{equation*}
so that 
\begin{equation*}
	\hrhono = \rpear(X, Y) = \frac{s_{xy}}{s_x\, s_y}.
\end{equation*}
%
%
and $ \breve X_i = (X_i - m_{10})/s_x$, $\breve Y_i = (Y_i - m_{01})/s_y$.
Denote the empirical moments of the $(\breve X_i, \breve Y_i)$
\begin{equation*}
	\breve{m}_{kl} = \frac1n\, \sum_{i=1}^n \breve X_i^k\, \breve Y_i^l,\qquad k,l \in \{0,1,\ldots, 4\}.
\end{equation*}
Then using $ \breve{m}_{20} = \breve{m}_{02} = 1$ we may write
\begin{equation}
	\hrhont = \rpear(\breve X^2,\breve Y^2) = \frac{\breve{m}_{22} - 1}{\big((\breve{m}_{40} - 1)\, (\breve{m}_{04} - 1)\big)^{1/2}},
\end{equation}
%
%
%

%
In terms of the empirical moments of $(X_i,Y_i)$, we have that
\begin{align}\label{eq:momentscentered}
	\breve{m}_{22} = \frac{1}{s_x^2\, s_y^2}\, \big( & m_{22} - 2\, m_{21}m_{01} - 2\, m_{12}m_{10} +  m_{20}m_{01}^2 + m_{10}^2 m_{02} \\
   & + 4\, m_{11}m_{10}m_{01} - 3 m_{10}^2m_{01}^2 \big)\nonumber\\
		\breve{m}_{40} = \frac{1}{s_x^4}\, \big( & m_{40} - 4 m_{30}m_{10} + 6 m_{20}\, m_{10}^2 - 3 m_{10}^4\big),
\end{align}
and similarly for $\breve{m}_{04}$.
Since both $\hrhono$ and $\hrhont$ use standardization, we may now assume that $X = \breve X$ and $Y= \breve Y$, 
%
that is $\Ex[X] = \Ex[Y] = 0$ and $\Ex[X^2] = \Ex[Y^2] = 1$ (but note that we still have to work with the empirical moments in the form \eqref{eq:momentscentered}). Recall the notation
$e_{kl} = \Ex[X^k Y^l]$ and $c_{kl,sr} = \Cov(X^k Y^l, X^s Y^r)$.
Then
%
\begin{align}\label{eq:cltmoments}
	\sqrt n \, \left(\begin{pmatrix}
		m_{10}\\
		m_{01}\\
		m_{20}\\
		m_{02}\\
		m_{11}\\
		m_{30}\\
		m_{03}\\
		m_{21}\\
		m_{12}\\
		m_{40}\\
		m_{04}\\
		m_{22}\\
	\end{pmatrix}  - \begin{pmatrix}
		0\\
		0\\
		1\\
		1\\
		\rho\\
		e_{30}\\
		e_{03}\\
		e_{21}\\
		e_{12}\\
		e_{40}\\
		e_{04}\\
		e_{22}\\
	\end{pmatrix}\right) \stackrel{\cL}{\to} \mathcal N_{12}\left(0, \Sigma\right), \qquad \Sigma_m =  \begin{pmatrix}
		c_{10,10} & c_{10,01} & \hdots & c_{10,22}\\
		c_{10,01} & c_{01,01} & \hdots & c_{01,22}\\
		\vdots & & & \vdots\\
		c_{10,22} & c_{01,22} & \hdots & c_{22,22}\\	
	\end{pmatrix}  .
\end{align}
%
Now consider the map $g : \R^{12} \to \R^6$ given by
$$g\big(		m_{10},m_{01},
m_{20},
m_{02},
m_{11},
m_{30},
m_{03},
m_{21},
m_{12},
\ldots,
m_{22} \big) = \big(s_x^2, s_y^2, s_{xy}, \breve{m}_{40}, \breve{m}_{04}, \breve{m}_{22} \big).$$
Its Jacobian $\D g$ at $(0,
0,
1,
1,
\rho,
e_{30},
e_{03},
e_{21},
e_{12},
e_{40},
e_{04},
e_{22})$ is computed as
\begin{align}\label{eq:matrixA}
A := 	& \D g ((0,
	0,
	1,
	1,
	\rho,
	e_{30},
	e_{03},
	e_{21},
	e_{12},
	e_{40},
	e_{04},
	e_{22}))\nonumber\\
	=  & \left(\begin{array}{cccccccccccc}
		0 & 0 & 1 & 0 & 0 & 0 & 0 & 0 & 0 & 0 & 0 & 0 \\
		0 & 0 & 0 & 1 &  0 & 0 & 0 & 0 & 0 & 0 & 0 & 0 \\
		0 & 0 & 0 & 0 &  1 & 0 & 0 & 0 & 0 & 0 & 0 & 0 \\
		- 4 e_{30}&  0 & - 2 e_{40} & 0 & 0 & 0 & 0 & 0 & 0 & 1 & 0 & 0 \\
		0 &  - 4 e_{03}&  0 & - 2 e_{04} & 0 & 0 & 0 & 0 & 0 & 0 & 1  & 0 \\
		- 2 e_{12} &  - 2 e_{21}&  -  e_{22} & - e_{22} & 0 & 0 & 0 & 0 & 0 & 0 & 0  & 1 \\
	\end{array}\right).
\end{align}
Furthermore, for the map
$$ h\big(s_x^2, s_y^2, s_{xy}, \breve{m}_{40}, \breve{m}_{04}, \breve{m}_{22} \big) = \Big(\frac{s_{xy}}{(s_x^2\, s_y^2)^{1/2}},\frac{\breve{m}_{22} -1}{\big((\breve{m}_{40} - 1)\,(\breve{m}_{04} - 1)\big)^{1/2}}  \Big),$$
observing
$$\rho_2 = \frac{e_{22} -1}{\big((e_{40} - 1)\,(e_{04} - 1)\big)^{1/2}}$$
and setting $a= g (0,
0,
1,
1,
\rho,
e_{30},
e_{03},
e_{21},
e_{12},
e_{40},
e_{04},
e_{22}),$
%
%
%
%
 we compute
\begin{align}\label{eq:matrixB}
	B:= \D h(a ) &=
	\left(\begin{array}{cccccc}
		- \rhono/2 & - \rhono/2 & 1 & 0 & 0 & 0  \\
		0 & 0 & 0 & - \frac{\rhont}{2\, (e_{40}-1)}  &  - \frac{\rhont}{2\, (e_{04}-1)}  & \frac{1}{\big( (e_{40} - 1)\,(e_{04} - 1)\big)^{1/2}}    \\
	\end{array}\right).
\end{align}
Observe that
$$ (\hrhono, \hrhont) = h\big( g\big(		m_{10},m_{01},
m_{20},
m_{02},
m_{11},
m_{30},
m_{03},
m_{21},
m_{12},
m_{40},
m_{04},
m_{22}  \big) \big),$$
as well as
$$ (\rhono, \rhont) = h\big( g\big(		0,
0,
1,
1,
\rho,
e_{30},
e_{03},
e_{21},
e_{12},
e_{40},
e_{04},
e_{22}
 \big) \big)$$
by the assumed standardization of $X$ and $Y$.
Asymptotic normality as claimed in Proposition 1 
follows from the $\Delta$ - method together with the form of the asymptotic covariance matrix as given in the appendix of the main paper follows from \eqref{eq:cltmoments}, \eqref{eq:matrixB} and the chain rule:
$$\D h\big( g(		0,
0,
1,
1,
\rho,
e_{30},
e_{03},
e_{21},
e_{12},
e_{40},
e_{04},
e_{22}
) \big) = B\, A  = M.$$

%

\end{proof}


%

\section{Proof of Lemmas \ref{lem:maxnormal} and \ref{lem:maxnormal1} }

\begin{proof} 
	To show  \eqref{eq:distrmax}, 
    note for the conditional distribution of $v$ given $u$ that
	$$ \cL(V | U = t) = \mathcal N ( \tau \, \sigma_2 \, t / \sigma_1 , \sigma_2^2\, (1- \tau^2)).$$
	Hence the formula follows from
	$$ \PP\big(Z \leq z \big) = \PP\big(U \leq z, V \leq z \big) = \int_{- \infty}^z\, \PP\big(V \leq z | U = t \big)\, \dd \PP_U(t).$$
	To show \eqref{eq:densermax}, 
    we use the formula
	\begin{equation}\label{eq:derivativegeneral}
		I(z)' = \int_{\alpha(z)}^{\beta(z)} \partial_z f(t,z)\, \dd t + \beta'(z)\, f(\beta(z),z) - \alpha'(z)\, f(\alpha(z),z)
	\end{equation}
	for the derivative of
	$$I(z) = \int_{\alpha(z)}^{\beta(z)}  f(t,z)\, \dd t.$$
	Apply this to \eqref{eq:distrmax}.
    The first term in \eqref{eq:densermax}
    then corresponds to the second in the derivative of $I(z)$. For the first term in $I(z)$, one shows that
	\begin{equation}\label{eq:formaludens}
		\varphi(t;0, \sigma_1^2)\, \varphi\Big(z; \tau \, \sigma_2 \, t / \sigma_1, \sigma_2\, \sqrt{1-\tau^2} \Big) =  \varphi(z;0, \sigma_2^2)\, \varphi\Big(t; \tau \, \sigma_1 \, z / \sigma_2, \sigma_1\, \sqrt{1-\tau^2} \Big),
	\end{equation}
	from which the second term in \eqref{eq:densermax}
    readily follows by integration.
\end{proof}
\begin{proof}[Proof of Lemma \ref{lem:maxnormal1}]
	%
	%
	Since
	\begin{align*}
		\PP\big(Z \leq z \big) & = 2\, \PP\big(0 \leq U \leq z, -z \leq V \leq z \big)\\
		& = 2\, \int_0^z\, \PP\big(-z \leq V \leq z | U = t \big)\, \dd \PP_U(t),
	\end{align*}
	\eqref{eq:distrmax1} follows readily from the conditional normal distribution of $V | U=t$.

	For \eqref{eq:densermax1}, use \eqref{eq:derivativegeneral} in \eqref{eq:distrmax1}, and observe \eqref{eq:formaludens}.
	
\end{proof}

\section{Additional simulations: Asymptotic confidence intervals} \label{sec-supp7}

Tables \ref{supp-tab3}-\ref{supp-tab6} show the empirical coverage probability and simulated mean length of the various confidence intervals defined in Subsection 4.2 for sample size $n=50$ and $n=800$, based on $10000$ replications.

\begin{table}[ht]
\centering
\begin{tabular}{lrrrrrr}
  \hline
distribution & $\text{ci}_{L,l}^{(p)}$ & $\text{ci}_{L,l}^{(p,c)}$ & $\text{ci}_{L,l}^{(b)}$ & $\text{ci}_{L,l}^{(b,c)}$ & $\text{ci}_{L}^{(b)}$ & $\text{ci}_{L}^{(b,c)}$ \\
  \hline
BVN(0) & 0.79 & 0.90 & 0.84 & 0.93 & 0.83 & 0.92 \\ 
  BVN(0.5) & 0.92 & 0.97 & 0.93 & 0.98 & 0.93 & 0.97 \\ 
  BVN(0.95) & 0.91 & 0.99 & 0.93 & 0.99 & 0.98 & 0.98 \\ 
  MN1 & 0.92 & 0.97 & 0.96 & 0.99 & 0.97 & 0.99 \\ 
  MN2 & 0.92 & 0.97 & 0.95 & 0.99 & 0.96 & 0.99 \\ 
  MN3 & 0.88 & 0.95 & 0.91 & 0.97 & 0.92 & 0.98 \\ 
  MN & 0.86 & 0.94 & 0.88 & 0.95 & 0.84 & 0.93 \\ \hline
  BVT5(0) &   -    &  -   &   -  &   -  & 0.97 & 0.99 \\ 
  BVT2(0) &   -    &  -   &   -  &   -  & 0.86 & 0.87 \\ 
  BVT1(0) &   -    &  -   &   -  &   -  & 0.82 & 0.82 \\ 
  BVT5(0.2) &   -    &  -   &   -  &   -  & 0.96 & 0.99 \\ 
  BVT2(0.2) &   -    &  -   &   -  &   -  & 0.88 & 0.88 \\ 
  BVT1(0.2) &   -    &  -   &   -  &   -  & 0.83 & 0.84 \\ \hline 
  UnifDisc & 0.97 & 0.98 & 0.98 & 0.99 & 0.98 & 1.00 \\ 
  UnifRhomb & 0.96 & 0.97 & 0.98 & 0.98 & 0.98 & 1.00 \\ 
  UnifTriangle & 0.93 & 0.97 & 0.93 & 0.98 & 0.92 & 0.97 \\ \hline
  RegLin1 & 0.91 & 0.99 & 0.92 & 0.99 & 0.94 & 0.98 \\ 
  RegLin2 & 0.92 & 0.98 & 0.93 & 0.98 & 0.93 & 0.98 \\ 
  RegQuad1 & 0.94 & 0.98 & 0.96 & 0.98 & 0.99 & 0.99 \\ 
  RegQuad2 & 0.96 & 1.00 & 0.97 & 1.00 & 0.97 & 1.00 \\ 
  RegTrig1 & 0.95 & 0.99 & 0.96 & 0.99 & 0.95 & 0.99 \\ 
  RegTrig2 & 0.95 & 0.99 & 0.95 & 0.99 & 0.95 & 0.99 \\ 
   \hline
\end{tabular}
\caption{Empirical coverage probability, $n=50$, $B=10000$. \label{supp-tab3}}
\end{table}

\begin{table}[ht]
\centering
\begin{tabular}{lrrrrrr}
  \hline
distribution & $\text{ci}_{L,l}^{(p)}$ & $\text{ci}_{L,l}^{(p,c)}$ & $\text{ci}_{L,l}^{(b)}$ & $\text{ci}_{L,l}^{(b,c)}$ & $\text{ci}_{L}^{(b)}$ & $\text{ci}_{L}^{(b,c)}$ \\
  \hline
  BVN(0)    & 0.39 & 0.40 & 0.41 & 0.41 & 0.41 & 0.41 \\ 
  BVN(0.5)  & 0.41 & 0.49 & 0.43 & 0.53 & 0.42 & 0.52 \\ 
  BVN(0.95) & 0.05 & 0.10 & 0.06 & 0.12 & 0.09 & 0.21 \\ 
  MN1       & 0.50 & 0.53 & 0.54 & 0.56 & 0.51 & 0.53 \\ 
  MN2       & 0.50 & 0.54 & 0.54 & 0.57 & 0.51 & 0.54 \\ 
  MN3       & 0.50 & 0.54 & 0.54 & 0.57 & 0.51 & 0.54 \\ 
  MN        & 0.42 & 0.43 & 0.43 & 0.44 & 0.41 & 0.41 \\ \hline
  BVT5(0)   &   -   &  -  &   -  &   -  & 0.49 & 0.50 \\ 
  BVT2(0)   &   -   &  -  &   -  &   -  & 0.56 & 0.61 \\ 
  BVT1(0)   &   -   &  -  &   -  &   -  & 0.53 & 0.63 \\ 
  BVT5(0.2) &   -   &  -  &   -  &   -  & 0.51 & 0.53 \\ 
  BVT2(0.2) &   -   &  -  &   -  &   -  & 0.56 & 0.61 \\ 
  BVT1(0.2) &   -   &  -  &   -  &   -  & 0.52 & 0.62 \\ \hline 
  UnifDisc  & 0.39 & 0.43 & 0.41 & 0.44 & 0.31 & 0.36 \\ 
  UnifRhomb & 0.29 & 0.33 & 0.31 & 0.36 & 0.25 & 0.31 \\ 
UnifTriangle& 0.37 & 0.42 & 0.38 & 0.44 & 0.44 & 0.53 \\ \hline
  RegLin1   & 0.26 & 0.38 & 0.27 & 0.40 & 0.28 & 0.43 \\ 
  RegLin2   & 0.37 & 0.45 & 0.38 & 0.48 & 0.38 & 0.48 \\ 
  RegQuad1  & 0.55 & 0.62 & 0.58 & 0.64 & 0.37 & 0.49 \\ 
  RegQuad2  & 0.49 & 0.52 & 0.51 & 0.53 & 0.43 & 0.47 \\ 
  RegTrig1  & 0.42 & 0.49 & 0.44 & 0.50 & 0.41 & 0.45 \\ 
  RegTrig2  & 0.44 & 0.48 & 0.45 & 0.49 & 0.43 & 0.46 \\ 
   \hline
\end{tabular}
\caption{Simulated mean length of confidence intervals, $n=50$, $B=10000$. \label{supp-tab4}}
\end{table}

\begin{table}
\centering
\begin{tabular}{lrrrrrr}
  \hline
distribution & $\text{ci}_{L,l}^{(p)}$ & $\text{ci}_{L,l}^{(p,c)}$ & $\text{ci}_{L,l}^{(b)}$ & $\text{ci}_{L,l}^{(b,c)}$ & $\text{ci}_{L}^{(b)}$ & $\text{ci}_{L}^{(b,c)}$ \\
  \hline
BVN(0) & 0.88 & 0.94 & 0.88 & 0.94 & 0.88 & 0.94 \\
  BVN(0.5) & 0.95 & 0.98 & 0.95 & 0.98 & 0.95 & 0.97 \\
  BVN(0.95) & 0.95 & 0.98 & 0.95 & 0.98 & 0.95 & 0.97 \\
  MN1 & 0.93 & 0.94 & 0.93 & 0.94 & 0.91 & 0.92 \\
  MN2 & 0.94 & 0.95 & 0.94 & 0.95 & 0.93 & 0.94 \\
  MN3 & 0.95 & 0.98 & 0.95 & 0.98 & 0.95 & 0.97 \\
  MN & 0.90 & 0.95 & 0.90 & 0.95 & 0.88 & 0.94 \\ \hline
  BVT5(0) &   -    &  -   &   -  &   -  & 0.90 & 0.90 \\
  BVT2(0) &   -    &  -   &   -  &   -  & 0.91 & 0.92 \\
  BVT1(0) &   -    &  -   &   -  &   -  & 0.92 & 0.93 \\
  BVT5(0.2) &  -   &  -   &   -  &   -  & 0.94 & 0.95 \\
  BVT2(0.2) &  -   &  -   &   -  &   -  & 0.91 & 0.92 \\
  BVT1(0.2) &  -   &  -   &   -  &   -  & 0.92 & 0.94 \\ \hline
  UnifDisc & 0.95 & 0.97 & 0.95 & 0.97 & 0.95 & 0.99 \\
  UnifDrhomb & 0.95 & 0.97 & 0.95 & 0.97 & 0.93 & 1.00 \\
  UnifTriangle & 0.95 & 0.97 & 0.95 & 0.97 & 0.94 & 0.97 \\ \hline
  RegLin1 & 0.95 & 0.98 & 0.95 & 0.98 & 0.94 & 0.98 \\
  RegLin2 & 0.95 & 0.98 & 0.95 & 0.98 & 0.95 & 0.98 \\
  RegQuad1 & 0.95 & 0.97 & 0.95 & 0.97 & 0.92 & 0.97 \\
  RegQuad2 & 0.95 & 0.97 & 0.95 & 0.97 & 0.93 & 0.96 \\
  RegTrig1 & 0.95 & 0.97 & 0.95 & 0.97 & 0.95 & 0.96 \\
  RegTrig2 & 0.95 & 0.97 & 0.95 & 0.97 & 0.95 & 0.96 \\
   \hline
\end{tabular}
\caption{Empirical coverage probability, $n=800$, $B=10000$. \label{supp-tab5}}
\end{table}

\begin{table}
\centering
\begin{tabular}{lrrrrrr}
  \hline
distribution & $\text{ci}_{L,l}^{(p)}$ & $\text{ci}_{L,l}^{(p,c)}$ & $\text{ci}_{L,l}^{(b)}$ & $\text{ci}_{L,l}^{(b,c)}$ & $\text{ci}_{L}^{(b)}$ & $\text{ci}_{L}^{(b,c)}$ \\
  \hline
BVN(0) & 0.11 & 0.11 & 0.11 & 0.11 & 0.11 & 0.11 \\
  BVN(0.5) & 0.10 & 0.15 & 0.10 & 0.15 & 0.10 & 0.14 \\
  BVN(0.95) & 0.01 & 0.03 & 0.01 & 0.03 & 0.01 & 0.04 \\
  MN1 & 0.19 & 0.20 & 0.19 & 0.20 & 0.18 & 0.19 \\
  MN2 & 0.19 & 0.20 & 0.19 & 0.20 & 0.17 & 0.18 \\
  MN3 & 0.17 & 0.19 & 0.17 & 0.19 & 0.16 & 0.18 \\
  MN & 0.11 & 0.11 & 0.11 & 0.11 & 0.10 & 0.11 \\ \hline
  BVT5(0) &  -    &  -   &   -  &   -  & 0.19 & 0.20 \\
  BVT2(0) &  -    &  -   &   -  &   -  & 0.18 & 0.20 \\
  BVT1(0) &  -    &  -   &   -  &   -  & 0.12 & 0.17 \\
  BVT5(0.2) &  -  &  -   &   -  &   -  & 0.18 & 0.19 \\
  BVT2(0.2) &  -  &  -   &   -  &   -  & 0.18 & 0.20 \\
  BVT1(0.2) &  -  &  -   &   -  &   -  & 0.12 & 0.16 \\ \hline
  UnifDisc & 0.09 & 0.10 & 0.09 & 0.10 & 0.04 & 0.07 \\
  UnifDrhomb & 0.06 & 0.07 & 0.06 & 0.07 & 0.03 & 0.06 \\
  UnifTriangle & 0.09 & 0.10 & 0.09 & 0.10 & 0.11 & 0.14 \\ \hline
  RegLin1 & 0.07 & 0.10 & 0.07 & 0.10 & 0.07 & 0.11 \\
  RegLin2 & 0.09 & 0.12 & 0.09 & 0.12 & 0.09 & 0.12 \\
  RegQuad1 & 0.14 & 0.16 & 0.14 & 0.16 & 0.08 & 0.12 \\
  RegQuad2 & 0.14 & 0.15 & 0.14 & 0.15 & 0.11 & 0.13 \\
  RegTrig1 & 0.11 & 0.12 & 0.11 & 0.12 & 0.10 & 0.10 \\
  RegTrig2 & 0.12 & 0.13 & 0.12 & 0.13 & 0.11 & 0.12 \\
   \hline
\end{tabular}
\caption{Simulated mean length of confidence intervals, $n=800$, $B=10000$. \label{supp-tab6}}
\end{table}

\section{Additional simulations: Testing for independence} \label{sec-supp6}

Tables \ref{supp-tab1} and \ref{supp-tab2} show the empirical power for the tests of independence described in Subsection 4.1  for sample size $n=25$ and $n=400$, based on $B=10000$ and $1000$ replications, respectively.

\begin{table}[ht]
\centering
\setlength{\tabcolsep}{4pt}
\begin{tabular}{lrrrrrrrrrrr}
  \hline
distribution & \hrpear & \hrspear & \hrhoamn & $\widehat{\rho_{\text{L(a)}}}$ & $\widehat{\rho_{\text{L(p)}}}$ & \hrace & \hdcor & \tstar & \hxi & \hhsic & \hbcor \\
  \hline
  BVN(0) & 0.05 & 0.05 & 0.05 & 0.05 & 0.05 & 0.05 & 0.05 & 0.05 & 0.05 & 0.05 & 0.05 \\ 
  BVN(0.5) & 0.74 & 0.68 & 0.65 & 0.61 & 0.62 & 0.19 & 0.69 & 0.61 & 0.29 & 0.41 & 0.43 \\ 
  BVN(0.95) & 1.00 & 1.00 & 1.00 & 1.00 & 1.00 & 1.00 & 1.00 & 1.00 & 1.00 & 1.00 & 1.00 \\ 
  MN1 & 0.10 & 0.07 & 0.21 & 0.17 & 0.17 & 0.13 & 0.09 & 0.06 & 0.06 & 0.08 & 0.10 \\ 
  MN2 & 0.18 & 0.14 & 0.25 & 0.20 & 0.21 & 0.12 & 0.16 & 0.12 & 0.08 & 0.12 & 0.14 \\ 
  MN3 & 0.27 & 0.22 & 0.30 & 0.26 & 0.27 & 0.12 & 0.24 & 0.19 & 0.11 & 0.16 & 0.18 \\ 
  MN & 0.05 & 0.05 & 0.05 & 0.04 & 0.05 & 0.05 & 0.05 & 0.05 & 0.05 & 0.05 & 0.05 \\ 
  BVT5(0) & 0.12 & 0.06 & 0.19 & 0.13 & 0.14 & 0.13 & 0.12 & 0.06 & 0.06 & 0.08 & 0.10 \\ 
  BVT2(0) & 0.25 & 0.08 & 0.42 & 0.35 & 0.36 & 0.31 & 0.38 & 0.07 & 0.08 & 0.22 & 0.29 \\ 
  BVT1(0) & 0.38 & 0.10 & 0.65 & 0.68 & 0.69 & 0.56 & 0.74 & 0.09 & 0.12 & 0.62 & 0.70 \\ 
  BVT5(0.2) & 0.22 & 0.16 & 0.25 & 0.20 & 0.21 & 0.14 & 0.21 & 0.14 & 0.09 & 0.12 & 0.14 \\ 
  BVT2(0.2) & 0.32 & 0.16 & 0.46 & 0.40 & 0.41 & 0.33 & 0.45 & 0.15 & 0.11 & 0.27 & 0.35 \\ 
  BVT1(0.2) & 0.42 & 0.17 & 0.67 & 0.70 & 0.70 & 0.58 & 0.76 & 0.17 & 0.16 & 0.66 & 0.74 \\ 
  UnifDisc & 0.02 & 0.02 & 0.13 & 0.01 & 0.01 & 0.06 & 0.04 & 0.04 & 0.05 & 0.08 & 0.05 \\ 
  UnifDrhomb & 0.00 & 0.01 & 0.31 & 0.06 & 0.07 & 0.14 & 0.03 & 0.04 & 0.08 & 0.21 & 0.09 \\ 
  UnifTriangle & 0.79 & 0.65 & 0.63 & 0.58 & 0.59 & 0.24 & 0.73 & 0.64 & 0.34 & 0.57 & 0.60 \\ 
  GARCH(2,1) & 0.11 & 0.07 & 0.18 & 0.20 & 0.20 & 0.11 & 0.13 & 0.06 & 0.07 & 0.16 & 0.19 \\ 
  RegLin1 & 0.99 & 0.98 & 0.98 & 0.96 & 0.96 & 0.58 & 0.98 & 0.97 & 0.74 & 0.90 & 0.91 \\ 
  RegLin2 & 0.84 & 0.80 & 0.77 & 0.71 & 0.72 & 0.24 & 0.81 & 0.76 & 0.38 & 0.59 & 0.60 \\ 
  RegQuad1 & 0.12 & 0.12 & 0.45 & 0.17 & 0.19 & 1.00 & 0.86 & 0.85 & 0.99 & 0.97 & 0.99 \\ 
  RegQuad2 & 0.09 & 0.10 & 0.18 & 0.10 & 0.10 & 0.66 & 0.43 & 0.37 & 0.69 & 0.65 & 0.67 \\ 
  RegTrig1 & 0.41 & 0.42 & 0.31 & 0.27 & 0.28 & 0.98 & 0.55 & 0.53 & 1.00 & 0.45 & 0.79 \\ 
  RegTrig2 & 0.29 & 0.31 & 0.21 & 0.21 & 0.21 & 0.59 & 0.34 & 0.35 & 0.83 & 0.26 & 0.37 \\ 
   \hline
\end{tabular}
\caption{Empirical power for several tests of independence with $n=25$, $B=10000$. \label{supp-tab1}}
\end{table}

\begin{table}[ht]
\centering
\setlength{\tabcolsep}{4pt}
\begin{tabular}{lrrrrrrrrrrr}
  \hline
distribution & \hrpear & \hrspear & \hrhoamn & $\widehat{\rho_{\text{L(a)}}}$ & $\widehat{\rho_{\text{L(p)}}}$ & \hrace & \hdcor & \tstar & \hxi & \hhsic & \hbcor \\
  \hline
BVN(0) & 0.05 & 0.05 & 0.05 & 0.06 & 0.05 & 0.05 & 0.05 & 0.05 & 0.05 & 0.04 & 0.05 \\ 
  BVN(0.5) & 1.00 & 1.00 & 1.00 & 1.00 & 1.00 & 1.00 & 1.00 & 1.00 & 1.00 & 1.00 & 1.00 \\ 
  BVN(0.95) & 1.00 & 1.00 & 1.00 & 1.00 & 1.00 & 1.00 & 1.00 & 1.00 & 1.00 & 1.00 & 1.00 \\ 
  MN1 & 0.12 & 0.09 & 0.98 & 0.98 & 0.98 & 0.89 & 0.24 & 0.10 & 0.09 & 0.62 & 0.75 \\ 
  MN2 & 0.88 & 0.89 & 1.00 & 1.00 & 1.00 & 0.92 & 0.94 & 0.89 & 0.19 & 0.92 & 0.96 \\ 
  MN3 & 1.00 & 1.00 & 1.00 & 1.00 & 1.00 & 0.99 & 1.00 & 1.00 & 0.37 & 0.99 & 0.99 \\ 
  MN & 0.08 & 0.07 & 0.06 & 0.05 & 0.04 & 0.06 & 0.08 & 0.07 & 0.05 & 0.08 & 0.08 \\ 
  BVT5(0) & 0.20 & 0.07 & 0.78 & 0.90 & 0.89 & 0.78 & 0.28 & 0.08 & 0.07 & 0.31 & 0.55 \\ 
  BVT2(0) & 0.67 & 0.08 & 0.96 & 1.00 & 1.00 & 1.00 & 1.00 & 0.20 & 0.15 & 1.00 & 1.00 \\ 
  BVT1(0) & 0.90 & 0.09 & 0.98 & 1.00 & 1.00 & 1.00 & 1.00 & 0.96 & 0.46 & 1.00 & 1.00 \\ 
  BVT5(0.2) & 0.92 & 0.95 & 0.98 & 1.00 & 0.99 & 0.92 & 0.99 & 0.94 & 0.21 & 0.86 & 0.93 \\ 
  BVT2(0.2) & 0.83 & 0.92 & 0.99 & 1.00 & 1.00 & 1.00 & 1.00 & 0.97 & 0.32 & 1.00 & 1.00 \\ 
  BVT1(0.2) & 0.92 & 0.88 & 0.99 & 1.00 & 1.00 & 1.00 & 1.00 & 1.00 & 0.63 & 1.00 & 1.00 \\ 
  UnifDisc & 0.01 & 0.02 & 1.00 & 1.00 & 1.00 & 0.99 & 0.21 & 0.16 & 0.12 & 0.97 & 0.98 \\ 
  UnifDrhomb & 0.00 & 0.01 & 1.00 & 1.00 & 1.00 & 1.00 & 1.00 & 0.98 & 0.33 & 1.00 & 1.00 \\ 
  UnifTriangle & 1.00 & 1.00 & 1.00 & 1.00 & 1.00 & 1.00 & 1.00 & 1.00 & 1.00 & 1.00 & 1.00 \\ 
  GARCH(2,1) & 0.31 & 0.08 & 1.00 & 1.00 & 1.00 & 0.99 & 0.95 & 0.29 & 0.17 & 0.99 & 1.00 \\ 
  RegLin1 & 1.00 & 1.00 & 1.00 & 1.00 & 1.00 & 1.00 & 1.00 & 1.00 & 1.00 & 1.00 & 1.00 \\ 
  RegLin2 & 1.00 & 1.00 & 1.00 & 1.00 & 1.00 & 1.00 & 1.00 & 1.00 & 1.00 & 1.00 & 1.00 \\ 
  RegQuad1 & 0.11 & 0.13 & 1.00 & 1.00 & 1.00 & 1.00 & 1.00 & 1.00 & 1.00 & 1.00 & 1.00 \\ 
  RegQuad2 & 0.08 & 0.09 & 0.99 & 1.00 & 1.00 & 1.00 & 1.00 & 1.00 & 1.00 & 1.00 & 1.00 \\ 
  RegTrig1 & 1.00 & 1.00 & 1.00 & 1.00 & 1.00 & 1.00 & 1.00 & 1.00 & 1.00 & 1.00 & 1.00 \\ 
  RegTrig2 & 1.00 & 1.00 & 1.00 & 1.00 & 1.00 & 1.00 & 1.00 & 1.00 & 1.00 & 1.00 & 1.00 \\ 
   \hline
\end{tabular}
\caption{Empirical power for several tests of independence with $n=400$, $B=1000$. \label{supp-tab2}}
\end{table}

\end{document}